\pdfoutput=1
\documentclass[11pt]{article}

\usepackage[T1]{fontenc}
\usepackage[utf8]{inputenc}
\usepackage{lmodern}
\usepackage[a4paper,margin=1in]{geometry}
\usepackage{microtype}
\usepackage{amsmath,amssymb,amsthm,mathtools}
\usepackage{aliascnt}
\usepackage{booktabs,longtable,array,multirow}
\usepackage{graphicx}
\usepackage{float}
\usepackage{enumitem}
\usepackage{xcolor}
\usepackage{natbib}
\usepackage{url}
\usepackage{xspace}
\usepackage{hyperref}
\usepackage[nameinlink,capitalise,noabbrev]{cleveref}
\usepackage{listings}
\usepackage{tikz}
\usetikzlibrary{arrows.meta,positioning,shapes.geometric,fit,calc}

\hypersetup{
  colorlinks=true,
  linkcolor=black,
  citecolor=black,
  urlcolor=black,
  pdftitle={Axient: Debt-Free Finality for Leveraged Binary Event Markets},
  pdfauthor={Maksym Nechepurenko},
  pdfsubject={A robust physically backed margin-layer mechanism with pre-finality debt extinguishment for leveraged binary event markets},
  pdfkeywords={prediction markets, event contracts, leverage, robust optimization, margin lending, conditional tokens, oracle disputes, liquidation, decentralized finance}
}
\newtheorem{definition}{Definition}[section]
\newaliascnt{assumption}{definition}
\newtheorem{assumption}[assumption]{Assumption}
\aliascntresetthe{assumption}
\newaliascnt{proposition}{definition}
\newtheorem{proposition}[proposition]{Proposition}
\aliascntresetthe{proposition}
\newaliascnt{theorem}{definition}
\newtheorem{theorem}[theorem]{Theorem}
\aliascntresetthe{theorem}
\newaliascnt{corollary}{definition}
\newtheorem{corollary}[corollary]{Corollary}
\aliascntresetthe{corollary}
\newaliascnt{lemma}{definition}
\newtheorem{lemma}[lemma]{Lemma}
\aliascntresetthe{lemma}
\newaliascnt{remark}{definition}

\aliascntresetthe{remark}
\newaliascnt{example}{definition}

\aliascntresetthe{example}

\crefname{definition}{definition}{definitions}
\Crefname{definition}{Definition}{Definitions}
\crefname{assumption}{assumption}{assumptions}
\Crefname{assumption}{Assumption}{Assumptions}
\crefname{proposition}{proposition}{propositions}
\Crefname{proposition}{Proposition}{Propositions}
\crefname{theorem}{theorem}{theorems}
\Crefname{theorem}{Theorem}{Theorems}
\crefname{corollary}{corollary}{corollaries}
\Crefname{corollary}{Corollary}{Corollaries}
\crefname{lemma}{lemma}{lemmas}
\Crefname{lemma}{Lemma}{Lemmas}
\crefname{remark}{remark}{remarks}
\Crefname{remark}{Remark}{Remarks}
\crefname{example}{example}{examples}
\Crefname{example}{Example}{Examples}

\newcommand{\Axient}{\textnormal{\textsc{Axient}}\xspace}

\newcommand{\E}{\mathbb{E}}
\newcommand{\Prob}{\mathbb{P}}
\newcommand{\F}{\mathcal{F}}

\newcommand{\Sset}{\mathcal{S}}
\newcommand{\Qset}{\mathcal{Q}}
\newcommand{\Uop}{\mathcal{U}^{\mathrm{op}}}
\newcommand{\Ustress}{\mathcal{U}^{\mathrm{stress}}}
\newcommand{\Ufail}{\mathcal{U}^{\mathrm{fail}}}
\newcommand{\dd}{\mathrm{d}}
\newcommand{\yes}{\mathrm{YES}}
\newcommand{\no}{\mathrm{NO}}
\newcommand{\hf}{\mathrm{HF}}
\newcommand{\state}[1]{\texttt{#1}}
\newcommand{\pospart}[1]{\left[#1\right]^+}

\lstdefinestyle{axientcode}{
  basicstyle=\ttfamily\small,
  frame=single,
  breaklines=true,
  columns=fullflexible,
  showstringspaces=false,
  keywordstyle=\bfseries,
  commentstyle=\itshape,
  xleftmargin=0.5em,
  xrightmargin=0.5em
}
\title{\textbf{Axient: Debt-Free Finality for Leveraged Binary Event Markets}}
\author{Maksym Nechepurenko\thanks{Founder and Director of Research, ForesightFlow, the Research Department of Devnull FZCO, Dubai, United Arab Emirates. Email: \href{mailto:maksym@devnull.ae}{maksym@devnull.ae}. Research profile: \url{https://www.foresightflow.org/}.}}
\date{July 13, 2026}

\begin{document}
\maketitle

\begin{abstract}
Leveraged event positions combine two risks that ordinary spot prediction markets keep separate: a loan must be repaid while the financed outcome claim can become non-tradable before its oracle payout is final. A post-trade accounting identity alone is not enough to control that risk, because the amount implied by a visible order book can differ from the amount that eventually matches, settles, and becomes available for debt repayment. This paper specifies \Axient, a physically backed margin layer for binary event markets that separates \emph{leverage maturity} from \emph{claim maturity} and makes the hard-flat decision under explicit execution uncertainty.

The model is defined on a filtered probability space and distinguishes four economic objects: quoted book proceeds, matched proceeds, settled proceeds, and final redemption. At a hard-flat decision time $u$, an execution policy is evaluated over a registered operating uncertainty set. The protocol sells the smallest admissible quantity whose \emph{lower settled-proceeds envelope} covers an upper bound on debt at the settlement horizon plus an explicit buffer. The realized minimum sale is then computed from confirmed fills as an audit quantity; it is not used circularly to choose an order before execution. We prove a robust ex-ante debt-clearing theorem, pathwise debt-extinguishment and debt-free-finality invariants, maximality of residual spot exposure on the realized path, payout-vector and dispute-duration invariance of lender credit-principal exposure, and a necessary-condition result showing that positive debt cannot be outcome-invariant when the held token may pay zero and no external collateral exists.

The paper also derives an exact book-dependent leverage envelope, relegating the familiar scalar recovery-ratio formula to a linear-execution benchmark; formalizes aggregate hard-flat capacity so that multiple positions cannot reuse the same top-of-book liquidity; and gives scenario-conditional reserve bounds using shared execution curves rather than position-by-position double counting. An impossibility theorem identifies the boundary: no backend-only mechanism with leverage above one can guarantee zero shortfall if actual market closure, signer control, settlement, or executable liquidity may fall outside the registered operating set.

No external dataset is required. A deterministic verifier evaluates finite step books, partial fills, settlement delay, robust envelopes, adversarial book transformations, shared-book liquidation, reserve allocation, zero liquidity, and payout vectors in $\{0,\tfrac12,1\}$. The operating set and the broader registered stress set are author-specified in this paper; empirical calibration and out-of-sample validation of these sets are reserved for a separate study. The contribution is a conditional mechanism-design result and reference implementation boundary, not a production-safety claim. \Axient removes terminal outcome and dispute duration from the lender's loan channel only after confirmed debt extinction; venue, custody, oracle, chain, and operational risks remain explicit.
\end{abstract}


\clearpage
\noindent\textit{Research and implementation disclosure.} \Axient is an active research-and-development initiative of the author.\footnote{A non-archival project page is maintained at \url{https://axient.app}.} This creates an interest in the mechanism's success. The paper therefore separates pathwise invariants, robust-set guarantees, scenario-conditional statements, implementation assumptions, and unvalidated deployment claims. It does not claim production safety or deployability beyond the conditions stated in each result.

\section{Introduction}
\label{sec:introduction}

\subsection{Motivation}

Prediction markets and perpetual-futures-style trading systems are converging, but their balance-sheet clocks are not. A conventional crypto perpetual assumes that the reference asset remains continuously tradable and that forced liquidation can, in principle, convert collateral into cash throughout the life of the position. A binary event claim is structurally different. Its quoted price lies in $[0,1]$; the venue may stop accepting orders before the event is final; settlement may pass through match, on-chain confirmation, proposal, challenge, dispute, and redemption; and the held outcome token may ultimately pay zero.

The first paper in the ForesightFlow Event-Linked Perpetuals programme formalized the non-portability of ordinary perpetual mechanics to bounded event underlyings and evaluated a resolution-aware engine on observed Polymarket paths \citep{nechepurenko2026resolutionaware}. Two negative findings motivate the present design. First, terminal jumps were economically large relative to ordinary leveraged collateral. Second, a staged halt reduced in-flight final-hour liquidations but did not reduce terminal-jump bad debt: the dominant loss channel lived in the financed balance sheet, not in the halt. The companion taxonomy, manipulation, and supply-side papers extend the design space, show how leverage changes manipulation and informed-trading rents, and document the importance and observability limits of hybrid-CLOB liquidity \citep{nechepurenko2026taxonomy,nechepurenko2026manipulation,nechepurenko2026fillside}.

This paper asks a narrower question:

\begin{quote}
\emph{Can a margin layer over an existing spot event venue provide meaningful pre-event leverage while terminating the lender's event-outcome exposure before the claim enters non-tradable finality states, and can the required sale be selected ex ante under changing order books and asynchronous settlement?}
\end{quote}

The first half of the answer is an accounting architecture: extinguish debt before finality and leave only a fully funded residual claim. The second half is an execution-control problem: the sale quantity must be selected before the future fills and settlement are known. The revision developed here treats these as distinct results rather than using a realized post-settlement curve as if it were an ex-ante order instruction.

\subsection{Two clocks and four execution objects}

A leveraged event position has two maturities:
\begin{align}
\text{leverage maturity} &= \sigma, \\
\text{claim maturity} &= \tau_f,
\end{align}
where $\sigma$ is the confirmed time at which the debt has been extinguished and $\tau_f\geq\sigma$ is the time at which the payout vector becomes final. A third time $\tau_r\geq\tau_f$ records confirmed redemption into cash. The interval $[\sigma,\tau_f]$ may contain delay, challenge, dispute, override, or void; the interval $[\tau_f,\tau_r]$ may contain redemption latency. Neither interval carries the extinguished loan.

The execution layer distinguishes four objects:
\begin{enumerate}[label=(\roman*)]
  \item the proceeds implied by the book observed when the hard-flat decision is made;
  \item the proceeds recorded by the matcher;
  \item the proceeds confirmed by the settlement layer and available for debt repayment; and
  \item the cash received when a final outcome token is redeemed.
\end{enumerate}
The first is a quote, the second is an execution record, the third is an accounting asset, and the fourth is final user liquidity. Conflating them creates a circular hard-flat rule and hides settlement risk.

\subsection{Canonical mechanism}

A user contributes collateral $C$ and acquires actual YES or NO outcome tokens using $C$ plus a separately accounted loan. Before the underlying venue becomes non-tradable, the mechanism enters reduce-only and then hard-flat. At decision time $u$, it selects the smallest admissible sale quantity whose conservative lower envelope of \emph{settled} proceeds covers the maximum debt expected over the settlement horizon plus an explicit buffer. The execution policy may use IOC orders, retries, RFQ, or another venue path; only confirmed proceeds reduce debt. Once debt is zero, any residual tokens become ordinary fully funded spot claims.

This operation is \emph{auto-deleverage to spot}. It does not require a separate public derivatives CLOB, a new outcome token, or a new final-resolution oracle. It does require control over the financed asset while debt is positive, authority to submit risk-reducing orders, an accounting-grade settlement signal, a double-entry ledger, and a policy for the case in which the operating uncertainty set is breached.

\subsection{What is and is not proved}

The paper separates three levels of statement.

\paragraph{Pathwise accounting invariants.}
If settled proceeds are applied to the loan and debt becomes zero, later payout and dispute duration do not recreate the extinguished receivable. These statements are exact but conditional.

\paragraph{Ex-ante robust execution guarantee.}
If the actual execution path belongs to a registered operating uncertainty set, the venue remains tradable for the registered horizon, the controller retains signing and collateral authority, and the lower settled-proceeds envelope is respected, then the planned sale clears debt by the horizon. This is stronger than the pathwise invariant because the sale is selected before execution.

\paragraph{Failure boundary.}
No backend-only mechanism with $L>1$ can guarantee zero shortfall over paths that include premature closure, complete disappearance of executable liquidity, loss of liquidation authority, or settlement failure outside the registered horizon, unless independent collateral or a third-party guarantee covers the gap.

Thus the claim is not that every position is always safe. The claim is that the mechanism defines a precise operating region in which terminal payout and dispute duration are removed from the lender's loan channel, and a precise failure region in which external capital or recovery rules are required.

\subsection{Contributions}

The paper makes eleven contributions.

\begin{enumerate}[label=\arabic*.]
  \item \textbf{A physically backed two-clock instrument.} Leveraged exposure is represented by real outcome tokens and a separate loan; leverage maturity is separated from claim and redemption maturity.

  \item \textbf{A stochastic execution setting.} The venue is modeled on a filtered probability space with scheduled and actual close times and separate decision, match, settlement, finality, and redemption times.

  \item \textbf{Three proceeds curves and one redemption value.} Quoted, matched, and settled proceeds are defined separately. Only settled proceeds can extinguish debt; final redemption is a later user-cash event.

  \item \textbf{An ex-ante robust auto-deleverage operator.} The planned sale is the minimum quantity certified by a lower settled-proceeds envelope and an upper debt-service envelope over a registered operating uncertainty set.

  \item \textbf{A realized audit minimum.} After settlement, the smallest quantity that actually sufficed is computed from confirmed fills. It supports residual-exposure maximality and overshoot analysis without circularly determining the earlier order.

  \item \textbf{Robust debt-clearing and debt-free-finality results.} Under stated control, execution, and settlement assumptions, debt is cleared by the registered horizon; after confirmed extinction, lender credit-principal exposure is invariant to payout vector and dispute duration.

  \item \textbf{A necessary-condition result.} If the held token can pay zero and no independent collateral or guarantee covers the loan, positive debt surviving into finality is incompatible with outcome-invariant lender safety.

  \item \textbf{An exact book-dependent leverage envelope.} Entry cost, robust exit proceeds, debt growth, fees, and buffers determine admissible leverage. The scalar recovery-ratio formula is retained only as a linear benchmark corollary, not a production sizing rule.

  \item \textbf{Aggregate liquidity discipline.} Multiple positions in the same risk bucket cannot each reuse the same top-of-book depth. We define shared-book feasibility, a lowest-$\Gamma_i$ deterministic execution priority, and reserve/OI limits from aggregate rather than position-by-position proceeds.

  \item \textbf{Operating-set and failure-set separation.} Robust guarantees are attached to a registered operating set; premature close, zero liquidity, signer failure, and settlement failure belong to an explicit failure set covered only by restrictions, guarantees, or reserves.

  \item \textbf{A deterministic verifier and implementation capability map.} The accompanying code checks the formal invariants on finite order books, robust scenarios, shared-book liquidation, and dispute states without external data, while the architecture section maps the assumptions to current hybrid-venue capabilities and trust boundaries.
\end{enumerate}

\subsection{Relationship to the ForesightFlow programme}

This is a standalone mechanism-design study, not an additional numbered paper in the four-paper Event-Linked Perpetuals series. It uses the series as foundation and addresses one of its negative findings through a different instrument architecture. The earlier framework carried synthetic leverage toward event resolution; \Axient instead terminates the loan before finality and preserves only fully funded residual directionality.

\subsection{Scope}

The base specification covers binary YES/NO markets, isolated margin, physically backed positions, and auto-deleverage to spot. It excludes cross-margin, portfolio offsets, multi-outcome markets, endogenous liquidity equilibrium, public permissionless lending, and production claims about any named venue. Current venue documentation is used only to map implementation capabilities; the formal results remain venue-agnostic.

No external dataset is used in this revision. Numerical examples are deterministic verification cases, not empirical calibration. A later empirical study can estimate the uncertainty set, settlement-horizon quantiles, and aggregate depth limits from venue data.

\subsection{Roadmap}

\Cref{sec:related} positions the work. \Cref{sec:setting} defines the stochastic venue, control assumptions, and maturity clocks. \Cref{sec:execution} separates quoted, matched, and settled proceeds. \Cref{sec:auto} defines the robust planned sale and realized audit minimum. \Cref{sec:safety} proves the principal invariants. \Cref{sec:risk} derives exact leverage and margin constraints. \Cref{sec:aggregate} treats shared liquidity and open-interest limits. \Cref{sec:finality} formalizes dispute and redemption. \Cref{sec:impossibility} separates operating and failure sets and gives reserve bounds. \Cref{sec:deterministic} reports deterministic verification. \Cref{sec:architecture,sec:security} map the model to implementation and trust. \Cref{sec:limitations} states what remains unproved.

\section{Related Work}
\label{sec:related}

The paper sits at the intersection of event-contract design, optimal execution, robust optimization, secured lending, and oracle-mediated settlement. The relevant literatures solve different parts of the problem; none, to our knowledge, combines pre-finality debt extinction with a robust settled-proceeds controller for physically backed event positions.

\subsection{Event-linked leverage and bounded-event microstructure}

The direct foundation is the ForesightFlow Event-Linked Perpetuals programme. The first paper establishes the bounded-support, terminal-collapse, and resolution-zone risks of synthetic event perpetuals and reports that halt-side mechanics do not remove terminal-jump bad debt \citep{nechepurenko2026resolutionaware}. The taxonomy paper identifies which structural properties carry to conditional, spread, basket, volatility, liquidity, rolling, and funding-only variants \citep{nechepurenko2026taxonomy}. The manipulation paper separates market-price manipulation from real-world outcome manipulation and shows how leverage changes both informed-trading rents and manipulation incentives \citep{nechepurenko2026manipulation}. The supply-side paper documents fill-side concentration and the structural limits of address-level quote-lifecycle attribution on a hybrid event CLOB \citep{nechepurenko2026fillside}.

The present paper changes the instrument rather than further tuning the synthetic engine. Its central object is a loan secured operationally by actual outcome tokens during the tradable interval, followed by debt extinction and residual spot ownership before oracle finality.

Prediction-market prices need not equal objective probabilities because risk preferences, capital constraints, fees, information, and market structure enter the quoted price \citep{manski2006interpreting,wolfers2004prediction}. The model therefore treats price as an execution input rather than a calibrated probability. Conditional-token systems provide the relevant physical payoff primitive: collateral is split into complementary claims and later redeemed against a payout vector \citep{gnosis2020conditional}. Optimistic-oracle systems add a proposal and challenge interval between the real-world event and final payout \citep{uma2026oracle}.

\subsection{Optimal execution and order-book uncertainty}

Large-order execution is classically a dynamic control problem rather than a static mark-to-market calculation. Bertsimas and Lo derive adaptive execution strategies under stochastic prices and impact \citep{bertsimas1998execution}, while Almgren and Chriss characterize the frontier between expected execution cost and execution risk \citep{almgren2001execution}. Limit-order-book models make the same point at finer granularity: order arrivals, cancellations, queue position, and price changes make future proceeds random even when the current book is fully observed \citep{cont2010orderbook}.

\Axient does not claim to solve the general optimal-execution problem. Its objective is narrower: certify a minimum quantity sufficient to repay debt over a finite hard-flat horizon. The distinction between quoted, matched, and settled proceeds is the mechanism-design counterpart of execution shortfall. The controller is conservative: it works with a lower settled-proceeds envelope rather than an expected-price objective.

\subsection{Robust and stochastic optimization}

Robust optimization replaces an assumed probability distribution with an explicit uncertainty set and asks for feasibility across every path in that set \citep{bertsimas2004price}. Stochastic programming instead uses a probability model, recourse, and distributional objectives or constraints \citep{shapiro2014stochastic}. Both are natural frameworks for hard-flat control.

This revision uses a robust-set statement because no venue-specific empirical distribution is required. At a decision time $u$, the operator registers an operating set containing admissible book changes, fees, partial fills, retries, and settlement delays. The lower settled-proceeds envelope is the worst outcome inside that set. Later empirical work may replace or complement it with a conditional quantile or chance constraint. The paper deliberately avoids adding an arbitrary Gaussian error to a cumulative order-book curve: such a representation can violate non-negativity, quantity monotonicity, lot structure, and the discrete failure states that dominate event-venue risk.

\subsection{Risk measures, margin, and secured lending}

Coherent-risk-measure and CVaR frameworks provide useful language for aggregating loss distributions and stress scenarios \citep{artzner1999coherent,rockafellar2000optimization}. The paper uses them as design references but does not claim that its maintenance buffer is a uniquely optimal coherent risk measure. The central credit result is more elementary: a claim that may pay zero cannot support outcome-invariant positive debt without external collateral or a guarantee.

The secured-lending analogy is incomplete unless collateral control is enforceable. A bookkeeping entry that labels outcome tokens as collateral does not prevent a borrower from withdrawing them, refusing to sign a liquidation order, or invoking an emergency path. Accordingly, the formal model includes execution-authority and collateral-nonescape assumptions. The implementation can satisfy them through a controlled subaccount, a venue-recognized margin vault, or another enforceable lien; these variants differ materially in trust.

\subsection{Hybrid venues and asynchronous settlement}

Hybrid event venues commonly combine off-chain matching with on-chain custody and settlement. Public PredictStreet documentation, for example, describes REST and WebSocket order-book access, EIP-712 signed orders, a match-to-settlement lifecycle, per-user outcome-token vaults, scheduled market-close fields, oracle challenges, delayed/void states, and asynchronous redemption \citep{predictstreet2026marketoverview,predictstreet2026websocket,predictstreet2026orderlifecycle,predictstreet2026vaults,predictstreet2026challenges,predictstreet2026voiddelay,predictstreet2026redeem}. Those primitives are sufficient to motivate the abstraction but do not automatically create a lender lien or delegated liquidation right.

The formal paper therefore does not treat a named platform as part of a theorem. A venue capability matrix in \Cref{sec:architecture} records which public interfaces support observation, which require a controlled signer, and which would require venue-side contract changes. This separation prevents a conditional mathematical result from being mistaken for a claim that a particular public API already supplies every enforcement primitive.

\subsection{Boundary with synthetic event perpetuals}

A synthetic event perpetual keeps a bilateral or pooled derivative obligation alive through a period in which the underlying event may no longer trade. A physically backed margin layer instead owns the outcome claim and can extinguish its loan before finality. The two structures may expose a similar user PnL before hard-flat, but their terminal balance sheets differ. The present paper analyzes only the latter and does not claim that its results automatically transfer to a cash-settled synthetic venue.

\section{Stochastic Setting, Instrument, and Control Assumptions}
\label{sec:setting}

\subsection{Probability space and information}

Work on a filtered probability space
\begin{equation}
(\Omega,\F,(\F_t)_{t\geq0},\Prob),
\label{eq:filtered-space}
\end{equation}
where $\F_t$ contains the order-book observations, venue status, matched trades, settlement confirmations, account balances, oracle messages, and chain events available to the mechanism by time $t$. All decision rules are adapted to $(\F_t)$. No assumption is made that order-book evolution is Gaussian, stationary, or independent of the mechanism's own trades.

\subsection{Binary event and payout vector}

Let $E$ be a binary event represented by complementary outcome claims $\yes$ and $\no$. The outcome-token system ultimately publishes a final payout vector
\begin{equation}
\boldsymbol{\pi}=(\pi_{\yes},\pi_{\no})\in[0,1]^2,
\qquad
\pi_{\yes}+\pi_{\no}=1.
\label{eq:payout-vector}
\end{equation}
The canonical states are $(1,0)$ and $(0,1)$. A void or invalid state may use another vector, including $(\tfrac12,\tfrac12)$, but the mechanism never hardcodes a neutral value. The authoritative input is the final vector actually recognized by the outcome-token contract or economically equivalent settlement system.

A user selects a held token $h\in\{\yes,\no\}$. Long-YES is $h=\yes$; short-YES is implemented as long-NO, $h=\no$, rather than as unsecured negative YES inventory.

\begin{definition}[Fully funded event claim]
A quantity $q\geq0$ of held token $h$ is fully funded if the outcome-token system has already locked the collateral required to pay $q\pi_h$ for every admissible payout vector and the holder has no borrowing obligation associated with that quantity.
\end{definition}

\begin{assumption}[Outcome collateralization]
\label{ass:collateralization}
Final redemption of outcome tokens is paid from collateral economically separate from the \Axient lender's loan receivable.
\end{assumption}

\subsection{Venue abstraction}

The underlying venue supplies:
\begin{enumerate}[label=(V\arabic*)]
  \item a market catalogue and observable market state;
  \item a quoted execution surface while the market is open;
  \item an order path that may include submission, matching, and asynchronous settlement;
  \item custody or accounting of acquired outcome tokens;
  \item a final payout vector after oracle finality; and
  \item redemption of outcome tokens into collateral.
\end{enumerate}
The venue may be a hybrid CLOB, RFQ system, AMM, or another mechanism. Order matching need not be on-chain. The model only requires that settlement confirmation and final payout are observable.

\subsection{Decision, execution, and finality times}

Let $t_0$ be the confirmed position-open time. The mechanism distinguishes:
\begin{itemize}
  \item $\sigma_R$: reduce-only transition;
  \item $u$: hard-flat decision and execution-start time;
  \item $\nu$: time of the last match used by the debt-clearing execution policy;
  \item $\sigma$: confirmed time at which settled proceeds have been applied and debt is zero;
  \item $T_c^{\mathrm{sched}}$: venue-published scheduled close time;
  \item $T_c^{\mathrm{act}}=\inf\{t:\text{risk-reducing orders are no longer accepted}\}$: actual close time;
  \item $\tau_p$: first provisional oracle proposal;
  \item $\tau_f$: final payout-vector time after any challenge or delay; and
  \item $\tau_r$: confirmed redemption time for the residual outcome token.
\end{itemize}

The intended successful ordering is
\begin{equation}
 t_0<\sigma_R\leq u\leq\nu\leq\sigma<T_c^{\mathrm{act}}\leq\tau_p\leq\tau_f\leq\tau_r.
\label{eq:successful-ordering}
\end{equation}
All times except the scheduled close may be random. If debt is never extinguished, set $\sigma=\infty$. If redemption is never confirmed, set $\tau_r=\infty$. A published $T_c^{\mathrm{sched}}$ is a planning input, not a guarantee that $T_c^{\mathrm{act}}$ cannot occur earlier.

\begin{definition}[Leverage maturity]
Leverage maturity is $\sigma$, the first confirmed time at which debt is zero after applying settled hard-flat proceeds. It is not order submission, match, or scheduled close.
\end{definition}

\begin{definition}[Claim and cash maturity]
Claim maturity is $\tau_f$, when the payout vector becomes final. Cash maturity is $\tau_r$, when redemption is confirmed. In general $\sigma\leq\tau_f\leq\tau_r$.
\end{definition}

\begin{figure}[H]
\centering
\resizebox{0.99\textwidth}{!}{%
\begin{tikzpicture}[
  node distance=0.65cm and 0.75cm,
  every node/.style={font=\small},
  box/.style={draw, rounded corners, minimum height=0.72cm, align=center, inner xsep=6pt},
  arrow/.style={-{Latex[length=2mm]}, thick},
  edge label/.style={font=\scriptsize, fill=white, inner sep=1pt}
]
\node[box] (open) {\state{OPEN}\\leveraged};
\node[box, right=of open] (reduce) {\state{REDUCE\_ONLY}};
\node[box, right=of reduce] (decision) {decision $u$\\robust sale cap};
\node[box, right=of decision] (match) {match $\nu$\\not final};
\node[box, right=of match] (settled) {settlement $\sigma$\\$D=0$};
\node[box, below=of settled] (pending) {\state{PENDING\_FINALITY}\\$\tau_p$};
\node[box, left=of pending] (dispute) {\state{DELAYED / DISPUTED}};
\node[box, right=of pending] (final) {final vector $\tau_f$};
\node[box, right=of final] (redeem) {redemption $\tau_r$};

\draw[arrow] (open) -- node[edge label, above] {$\sigma_R$} (reduce);
\draw[arrow] (reduce) -- (decision);
\draw[arrow] (decision) -- (match);
\draw[arrow] (match) -- (settled);
\draw[arrow] (settled) -- (pending);
\draw[arrow] (pending) -- (dispute);
\draw[arrow] (dispute.east) -- (final.west);
\draw[arrow] (pending) -- (final);
\draw[arrow] (final) -- (redeem);
\end{tikzpicture}%
}
\caption{The \Axient clocks. The sale decision precedes match and settlement. The transition into finality states is permitted only after settlement has extinguished debt. Final payout and redemption are later, distinct events.}
\label{fig:lifecycle}
\end{figure}
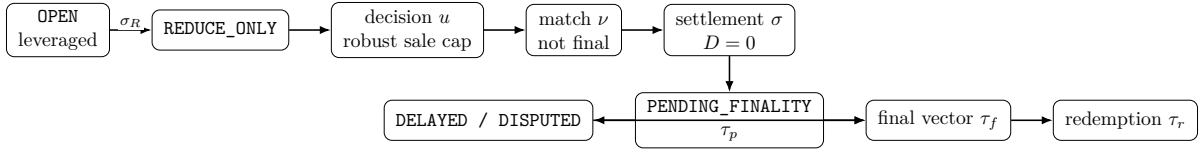

\subsection{Position construction}

A user contributes collateral $C>0$ and requests leverage $L\geq1$. Gross acquisition budget and initial principal are
\begin{equation}
N=LC,
\qquad
D_0=(L-1)C.
\label{eq:position-budget-debt}
\end{equation}
The confirmed acquisition cost determines token quantity $q_0$ and any residual cash $K_0\geq0$ as defined in \Cref{sec:execution}. Debt accrues until $\sigma$ and never after it.

This is not a peer-to-peer synthetic long-short obligation. The balance sheet contains an identifiable spot asset $q_t$, free cash $K_t$, and loan receivable $D_t$.

\subsection{Execution authority and collateral control}

The ex-ante guarantee requires more than a price model.

\begin{assumption}[Risk-reducing execution authority]
\label{ass:authority}
While $D_t>0$, the \Axient controller or an enforceable on-chain policy can submit, replace, cancel, and settle permitted risk-reducing orders for the financed position without requiring a new discretionary signature from the borrower at the time of hard-flat.
\end{assumption}

\begin{assumption}[Collateral non-escape and loan priority]
\label{ass:lien}
While $D_t>0$, the borrower cannot withdraw, transfer, re-pledge, or otherwise remove the financed outcome tokens and dedicated cash from the execution policy. Settled sale proceeds are applied to the loan before residual value is released.
\end{assumption}

These assumptions may be implemented by a controlled subaccount, an MPC signer with restrictive policy, or a venue-recognized margin vault with a lender lien and delegated liquidation role. A standard self-custodial wallet in which the user alone can sign every sale does not satisfy \Cref{ass:authority} unless a prior enforceable delegation exists.

\begin{assumption}[No rehypothecation after conversion]
\label{ass:no-rehyp}
After debt is repaid and residual tokens are classified as fully funded spot, they are not attached to the extinguished loan or pledged to a new obligation without a new explicit transaction.
\end{assumption}

\subsection{Physical backing and custody are distinct axes}

Physical backing means the user position is represented by real, collateralized outcome tokens. Non-custodial control means the user alone controls the keys. A controlled subaccount can be physically backed but custodial; a user-controlled vault can be non-custodial but unsuitable for automatic hard-flat if the user can refuse the required signature. The paper keeps these properties separate and does not use ``physically backed'' as a synonym for ``permissionless'' or ``trustless.''

\subsection{Canonical policies}

The base mechanism adopts:
\begin{enumerate}[label=(P\arabic*)]
  \item isolated margin by market and risk bucket;
  \item no risk-increasing orders after $\sigma_R$;
  \item robust ex-ante hard-flat at $u<T_c^{\mathrm{sched}}$;
  \item no transition into finality states while ordinary debt remains positive;
  \item auto-deleverage to spot as the default, with full close optional;
  \item no interest after debt extinction;
  \item no final accounting from a provisional proposal;
  \item binary markets only in the base specification;
  \item actual payout-vector ingestion; and
  \item aggregate depth and OI limits that prevent reuse of the same liquidity.
\end{enumerate}

\subsection{Operating, stress, and failure paths}

At each hard-flat decision time $u$, the mechanism freezes a nested uncertainty taxonomy over a finite horizon $\Delta$:
\begin{equation}
\Uop_u(\Delta)
\subseteq
\Ustress_u(\Delta)
\subseteq
\Omega,
\qquad
\Ufail_u(\Delta)
=
\Omega\setminus\Uop_u(\Delta).
\label{eq:operating-stress-failure-sets}
\end{equation}
The \emph{operating set} $\Uop_u(\Delta)$ contains adverse but bounded book changes, partial fills, fees, and settlement delays for which the controller claims the robust debt-clearing certificate. The broader \emph{registered stress set} $\Ustress_u(\Delta)$ adds named adversarial book-response, control-degradation, correlated-close, and settlement-failure scenarios used for validation, reserve sizing, and recovery design. Paths in $\Ustress_u(\Delta)\setminus\Uop_u(\Delta)$ are deliberately tested but are not covered by the robust certificate. The full failure set $\Ufail_u(\Delta)$ contains every path outside the certificate, including paths not enumerated in the finite stress registry.

Formal guarantees apply only to $\Uop$. Stress and reserve statements are conditional on their registered scenario families; they do not retroactively enlarge the operating set. The distinction is central: a mechanism cannot declare a path ``covered'' merely because external capital absorbs the loss after execution fails. In this paper both $\Uop$ and $\Ustress$ are author-specified deterministic registries. Their empirical calibration, target coverage, and out-of-sample validation are separate research tasks.

\section{Quoted, Matched, and Settled Execution Value}
\label{sec:execution}

The credit layer is built on cash that can actually repay debt. A displayed midpoint is informative; a visible bid ladder is an execution quote; a match is an expected transfer; only settled collateral is an accounting asset. This section defines the four objects separately.

\subsection{Confirmed acquisition cost}

Let $A_{t_0}^{\mathrm{set}}(q)$ be the total confirmed collateral cost of acquiring $q\geq0$ units of held token $h$ at entry, including price impact, venue fees, protocol fees charged at entry, and settlement costs allocated to the trade.

\begin{assumption}[Acquisition-cost regularity]
\label{ass:acquisition}
On the admissible entry lot set $\Qset_{t_0}^{\mathrm{buy}}$, $A_{t_0}^{\mathrm{set}}(0)=0$ and $A_{t_0}^{\mathrm{set}}(q)$ is non-decreasing. Convexity is permitted but not required for the finite-lot results.
\end{assumption}

Given gross budget $N=LC$, confirmed quantity and residual cash are
\begin{align}
q_0(L)&=\max\left\{q\in\Qset_{t_0}^{\mathrm{buy}}:A_{t_0}^{\mathrm{set}}(q)\leq LC\right\},
\label{eq:entry-quantity}\\
K_0(L)&=LC-A_{t_0}^{\mathrm{set}}(q_0(L))\geq0.
\label{eq:entry-cash}
\end{align}
A continuous model replaces the maximum with a supremum. The base mechanism activates the loan only after the acquisition has settled; entry settlement failure therefore creates no open financed position.

\subsection{Book-implied proceeds}

At hard-flat decision time $u$, let $\mathcal O_u$ denote the observed book snapshot and let
\begin{equation}
B_u^{\mathrm{book}}(x;\mathcal O_u)
\label{eq:book-proceeds}
\end{equation}
be the net proceeds implied by immediately consuming the visible bid ladder for $x$ units, after fees and deterministic charges known at $u$. This is an $\F_u$-measurable quote. It is not a guarantee because orders can be cancelled, partially filled, repriced, rejected, or fail to settle.

For a finite CLOB, $B_u^{\mathrm{book}}$ is calculated by walking bid levels in price-time order. Fixed fees, minimum charges, and lot constraints may create discontinuities. The analysis therefore requires only monotonicity on a finite admissible set, not global concavity or continuity.

\subsection{Registered adversarial book transformations}

Let the observed bid ladder be $\mathcal L_u=\{(v_j,p_j)\}_{j=1}^{J}$ with $p_1\geq\cdots\geq p_J\geq0$, and let $m_u$ be the reference mid used for stress construction. For $k\in\{0,\ldots,J\}$, depth-haircut fraction $h\in[0,1]$, and spread-expansion factor $s\geq1$, define
\begin{equation}
T_{k,h,s}(\mathcal L_u)
=
\left\{
\left(
(1-h)v_{j+k},
\pospart{m_u-s(m_u-p_{j+k})}
\right)
:\ j=1,\ldots,J-k
\right\},
\label{eq:book-transform}
\end{equation}
with zero-quantity levels removed. The transform removes the best $k$ levels, removes fraction $h$ of quantity from every remaining level, and expands each remaining bid's distance from the reference mid by factor $s$. The identity is $T_{0,0,1}$; $k=J$ or $h=1$ produces an empty executable book. Because the price map is increasing in $p_j$, bid ordering is preserved.

The transform family is a deterministic stress primitive, not an equilibrium model of market-maker behavior. Selected transforms may be included in $\Uop$ only if the operating-set construction and empirical coverage target justify them. More severe transforms belong to $\Ustress\setminus\Uop$ and test the mechanism's failure classification and reserve waterfall. Adaptive market-maker behavior can be richer than any finite $T_{k,h,s}$ family; complete instantaneous withdrawal remains an explicit impossibility path.

\subsection{Execution policy and stochastic proceeds}

Let $\Pi_u$ be an $\F_u$-measurable hard-flat execution policy. It specifies permissible order types, price limits, retry timing, cancellation rules, and a maximum requested sale quantity $x$. For a future path $\omega$ over $[u,u+\Delta]$, define the proceeds associated with running the policy against cumulative cap $x$:

\begin{align}
B_{u,\Delta}^{\mathrm{match},\Pi}(x,\omega)
&=\text{net proceeds recorded by matched fills by }u+\Delta,\\
B_{u,\Delta}^{\mathrm{set},\Pi}(x,\omega)
&=\text{net proceeds confirmed and available by }u+\Delta.
\label{eq:match-set-proceeds}
\end{align}

The second quantity excludes fills that remain pending or fail. Both functions depend on the policy and path even when the current snapshot is fixed. They are capacity curves indexed by the permitted cumulative cap: a live controller may stop at an earlier confirmed prefix once its selected debt or debt-plus-buffer target is reached. In general,
\begin{equation}
B_{u,\Delta}^{\mathrm{set},\Pi}(x,\omega)
\leq
B_{u,\Delta}^{\mathrm{match},\Pi}(x,\omega),
\label{eq:settled-below-matched}
\end{equation}
with strict inequality during pending settlement or after a failed leg.

\begin{assumption}[Settlement observability and finality]
\label{ass:settlement}
The mechanism can identify which matched fills have become final under the venue's settlement rules. A fill reduces debt only after its associated collateral receipt is confirmed and cannot be reversed by an ordinary application retry.
\end{assumption}

Chain reorganization beyond the selected confirmation rule, contract exploit, or venue insolvency is an infrastructure failure rather than ordinary settlement latency.

\subsection{Operating lower envelope}

Fix a registered operating set $\Uop_u(\Delta)$ of execution paths and a policy $\Pi_u$. The robust lower settled-proceeds envelope is
\begin{equation}
\underline B_{u,\Delta}^{\Pi}(x)
=
\inf_{\omega\in\Uop_u(\Delta)}
B_{u,\Delta}^{\mathrm{set},\Pi}(x,\omega).
\label{eq:lower-proceeds-envelope}
\end{equation}
It is computed only on the admissible quantity set $\Qset_u(q)$. The set may encode spread widening, depth cancellation, partial fills, fee changes within stated bounds, retry costs, and settlement delays up to $\Delta$. It must not silently include paths for which the execution policy has no authority or the market has already closed; those belong to $\Ufail$.

The lower envelope is not an expectation. It is a certificate relative to a named uncertainty set. If a probabilistic model is later available, a conditional quantile can be reported separately:
\begin{equation}
\underline B_{u,\Delta}^{(\alpha)}(x)
=
Q_{\alpha}\!\left(
B_{u,\Delta}^{\mathrm{set},\Pi}(x)\mid\F_u
\right),
\label{eq:quantile-proceeds}
\end{equation}
but the deterministic results in this paper do not require estimating it.

\subsection{Realized audit curve}

After the execution sequence has settled, order the confirmed sale fills by their policy execution order. For a realized path $\omega$, let
\begin{equation}
B_{\sigma}^{\mathrm{audit}}(x,\omega)
\label{eq:audit-proceeds}
\end{equation}
be the cumulative settled proceeds attributable to the first $x$ units of the realized fill sequence, with pro-rata allocation inside a marginal fill if the venue supports divisible quantities. This is a post-settlement accounting object. It answers how many of the tokens actually sold were needed to repay the debt; it does not determine the earlier order quantity.

\begin{assumption}[Audit-curve monotonicity]
\label{ass:audit-monotone}
On the realized admissible cumulative-fill set, $B_{\sigma}^{\mathrm{audit}}(0)=0$ and the curve is non-decreasing. No continuity or concavity is required.
\end{assumption}

\subsection{Outstanding debt and debt-service envelope}

Fix hard-flat decision time $u$. Let $P_{u,t}$ be cumulative confirmed cash applied to the loan after $u$, including dedicated account cash and settled sale proceeds. Before extinction, the outstanding balance follows the debt-first accounting equation
\begin{equation}
D_t
=
\pospart{
D_u
+
\int_u^t r_sD_s\,\dd s
+
C_{u,t}^{\mathrm{debt}}
-
P_{u,t}
},
\qquad t\geq u,
\label{eq:debt-growth}
\end{equation}
where $r_t\geq0$ is the borrowing rate and $C_{u,t}^{\mathrm{debt}}$ contains charges contractually added after $u$. The positive-part convention prevents overpayment from manufacturing negative debt. New borrowing is prohibited after reduce-only.

For a realized path and policy, define the cumulative debt-service requirement through $t$ by
\begin{equation}
H_{u,t}^{\Pi}(\omega)
=
P_{u,t}^{\Pi}(\omega)+D_t^{\Pi}(\omega).
\label{eq:debt-service-requirement}
\end{equation}
This is the total cash already applied after $u$ plus the amount still required to extinguish the receivable at $t$. It remains well-defined when repayments occur in several settlement batches and automatically incorporates the effect of payment timing on accrued interest. If debt was extinguished earlier, $D_t=0$ and $H_{u,t}^{\Pi}$ equals cumulative debt service paid through extinction.

With policy $\Pi$ fixed, define the operating-set debt-service upper envelope
\begin{equation}
\overline H_{u,\Delta}
=
\sup_{\omega\in\Uop_u(\Delta)}
H_{u,u+\Delta}^{\Pi}(\omega).
\label{eq:debt-service-upper-envelope}
\end{equation}
A practical implementation may bound this quantity using maximum borrow rates, retry and settlement charges, and the latest permitted payment times. Using $D_u$ alone understates the cash requirement when settlement is delayed. A no-payment accrued balance is a conservative special case when early repayment cannot increase interest or charges.

The account may also hold dedicated cash $K_u\geq0$. Cash that is pending, withdrawable by another party, already counted in $P_{u,t}$, or legally unavailable to the loan is excluded from any additional coverage term.

\subsection{Executable equity and health factor}

For ordinary monitoring at time $t$, let $\underline V_t(q)$ be a conservative liquidation value derived from the current registered execution set. Executable equity and health factor are
\begin{align}
E_t(q)&=K_t+\underline V_t(q)-D_t,
\label{eq:equity}\\
\hf_t(q)&=\frac{K_t+\underline V_t(q)}{D_t+M_t(q)},
\label{eq:health-factor}
\end{align}
where $M_t(q)\geq0$ is a maintenance buffer. A position is inside its maintenance region when $\hf_t>1$. At hard-flat, the controller targets debt extinction rather than preservation of the leveraged state and may consume the maintenance buffer, subject to a separate hard-flat buffer $m_u$.

A generic decomposition is
\begin{equation}
M_t=M_t^{\mathrm{exec}}+M_t^{\mathrm{lat}}+M_t^{\mathrm{conc}}+M_t^{\mathrm{ops}}+M_t^{\mathrm{auth}},
\label{eq:buffer-decomp}
\end{equation}
where the final term reflects execution-authority and collateral-control risk. If authority is not enforceable, $M_t^{\mathrm{auth}}$ cannot repair the mechanism; new leverage should be disabled.

\subsection{Double-entry accounting boundary}

The minimum chart of accounts contains venue cash, outcome tokens, pending settlement, loan receivable, user cash, user spot claim, capital-provider payable, reserves, fees, and operational loss. A matched fill posts only to pending accounts. Settlement moves value into confirmed venue cash or token inventory. A failed settlement is reversed by compensating entries rather than history mutation.

\begin{assumption}[Ledger discipline]
\label{ass:ledger}
Every journal entry is balanced, immutable after posting, and associated with a unique external or internal event identifier. The same economic event cannot post twice under the same identifier.
\end{assumption}

\subsection{Why book and midpoint are insufficient on their own}

For quantity $q$ and midpoint $m_u$, the displayed mark value $qm_u$ can exceed the visible book proceeds $B_u^{\mathrm{book}}(q)$, which can in turn exceed later settled proceeds:
\begin{equation}
qm_u
\geq B_u^{\mathrm{book}}(q)
\quad\text{need not imply}\quad
B_{u,\Delta}^{\mathrm{set},\Pi}(q,\omega)
\geq D_{u+\Delta}.
\label{eq:valuation-chain}
\end{equation}
The first gap is book walk; the second is execution and settlement uncertainty. A credit policy that stops at the midpoint or current snapshot can therefore approve a position that appears overcollateralized yet cannot repay debt on the realized path.

\section{Ex-Ante Robust Auto-Deleverage}
\label{sec:auto}

The hard-flat controller must choose a sale quantity before future fills are known. This section separates the \emph{planned robust sale} from the \emph{realized audit minimum}.

\subsection{Admissible quantities}

Let a position at decision time $u$ contain settled quantity $q$, dedicated cash $K_u$, and debt $D_u>0$. Let
\begin{equation}
\Qset_u(q)=\{0=y_0<y_1<\dots<y_m\leq q\}
\label{eq:admissible-grid}
\end{equation}
be the finite cumulative sale grid generated by token precision, venue lot size, order limits, and risk policy. A continuous model may use $[0,q]$, but the finite formulation matches actual order books and requires weaker regularity.

Let $m_u\geq0$ be a dedicated hard-flat buffer. It is not ordinary maintenance margin; it covers registered residual uncertainty and may remain as user or reserve cash after debt extinction according to policy.

\subsection{Robust feasible set and planned sale}

\begin{definition}[Robust debt-clearing feasible set]
For execution policy $\Pi_u$, horizon $\Delta$, and operating set $\Uop_u(\Delta)$, define
\begin{equation}
\widehat{\mathcal F}_{u,\Delta}(q)
=
\left\{
x\in\Qset_u(q):
K_u+\underline B_{u,\Delta}^{\Pi}(x)
\geq
\overline H_{u,\Delta}+m_u
\right\}.
\label{eq:robust-feasible-set}
\end{equation}
\end{definition}

\begin{definition}[Planned robust sale]
If $\widehat{\mathcal F}_{u,\Delta}(q)\neq\varnothing$, the planned sale cap is
\begin{equation}
\widehat x_u
=
\min\widehat{\mathcal F}_{u,\Delta}(q).
\label{eq:planned-sale}
\end{equation}
If the set is empty, the position is not certified for ordinary hard-flat and enters a named shortfall or backstop path.
\end{definition}

The quantity $\widehat x_u$ is chosen from information available at $u$. It is the minimum quantity certified by the registered uncertainty set, not the minimum quantity that will turn out to have been necessary after settlement.

\begin{proposition}[Existence on a finite lot grid]
\label{prop:planned-existence}
The planned sale exists if and only if
\begin{equation}
K_u+\underline B_{u,\Delta}^{\Pi}(y_m)
\geq
\overline H_{u,\Delta}+m_u.
\label{eq:planned-feasibility}
\end{equation}
When it exists, it is unique as the minimum element of a finite ordered set.
\end{proposition}

\begin{proof}
If the inequality holds, $y_m$ belongs to the feasible set, so the set is non-empty and has a unique minimum. If it fails, monotonicity of the lower envelope on the ordered grid implies no smaller quantity can satisfy the threshold.
\end{proof}

\begin{proposition}[Certification minimality]
\label{prop:certification-minimality}
No $x\in\Qset_u(q)$ with $x<\widehat x_u$ is certified to clear debt under the same operating set, policy, horizon, debt bound, and buffer.
\end{proposition}

\begin{proof}
By definition of the minimum, every smaller admissible quantity fails the robust inequality in \eqref{eq:robust-feasible-set}. The proposition does not claim that every smaller quantity fails on every realized path; it states that the same robust certificate cannot support it.
\end{proof}

\subsection{Realized audit minimum}

Suppose the policy sells a cumulative quantity $x^{\mathrm{sold}}$ and settlement completes at $\sigma<\infty$. Let $D_{\sigma^-}$ be debt immediately before applying the confirmed sale cash. Define the realized feasible set
\begin{equation}
\mathcal F_{\sigma}^{\mathrm{audit}}
=
\left\{
x\in\Qset_{\sigma}(x^{\mathrm{sold}}):
K_u+B_{\sigma}^{\mathrm{audit}}(x)
\geq D_{\sigma^-}
\right\}.
\label{eq:audit-feasible-set}
\end{equation}

\begin{definition}[Realized audit minimum]
If $\mathcal F_{\sigma}^{\mathrm{audit}}\neq\varnothing$, define
\begin{equation}
x_{\sigma}^{\ast}
=
\min\mathcal F_{\sigma}^{\mathrm{audit}}.
\label{eq:audit-minimum}
\end{equation}
\end{definition}

The audit minimum is useful for explaining realized residual exposure, measuring conservative overshoot, and reconciling the controller. It is never substituted for $\widehat x_u$ in the ex-ante rule.

\begin{proposition}[Audit existence]
\label{prop:audit-existence}
The realized audit minimum exists exactly when the settled execution sequence plus dedicated cash covers debt:
\begin{equation}
K_u+B_{\sigma}^{\mathrm{audit}}(x^{\mathrm{sold}})
\geq D_{\sigma^-}.
\label{eq:audit-coverage}
\end{equation}
\end{proposition}

\begin{proof}
Identical to \Cref{prop:planned-existence}, using the realized finite cumulative-fill grid and \Cref{ass:audit-monotone}.
\end{proof}

\subsection{Relationship between planned and realized quantities}

When the realized path lies in the operating set and the controller stops as soon as confirmed proceeds clear debt, the planned cap bounds the realized requirement.

\begin{proposition}[Planned cap dominates realized requirement]
\label{prop:planned-dominates-realized}
Suppose $\omega\in\Uop_u(\Delta)$, $T_c^{\mathrm{act}}(\omega)>u+\Delta$, \Cref{ass:authority,ass:lien,ass:settlement} hold, no new debt is created, and the controller follows $\Pi_u$ by submitting risk-reducing sales until debt is zero or cumulative settled quantity reaches $\widehat x_u$. Assume every submitted fill used by the policy settles by $u+\Delta$, and unsent or unmatched remainder may be cancelled after confirmed clearance. If $\widehat x_u$ exists, then the audit minimum exists and
\begin{equation}
x_{\sigma}^{\ast}\leq \widehat x_u.
\label{eq:audit-below-plan}
\end{equation}
\end{proposition}

\begin{proof}
The robust inequality gives
$K_u+\underline B_{u,\Delta}^{\Pi}(\widehat x_u)\geq\overline H_{u,\Delta}+m_u$.
Because $\omega$ lies in the operating set, the settled-proceeds capacity at the planned cap is at least the lower envelope, while realized debt is at most the upper envelope. Therefore the controller must reach debt clearance no later than the certified cap. The minimum realized sufficient quantity cannot exceed that cap.
\end{proof}

\subsection{Residual-exposure maximality}

For a realized path, any debt-clearing sale $x$ leaves residual quantity $q-x$. The audit minimum therefore preserves the greatest physically funded event exposure compatible with debt repayment.

\begin{proposition}[Residual-exposure maximality]
\label{prop:residual-maximality}
If $x_{\sigma}^{\ast}$ exists, then for every realized debt-clearing quantity $x$,
\begin{equation}
q-x_{\sigma}^{\ast}\geq q-x.
\label{eq:residual-maximality}
\end{equation}
\end{proposition}

\begin{proof}
Every debt-clearing $x$ belongs to $\mathcal F_{\sigma}^{\mathrm{audit}}$, so $x\geq x_{\sigma}^{\ast}$ by minimality.
\end{proof}

This is a pathwise result. Robust planning may deliberately target more than the eventual audit minimum because it must protect against adverse paths that did not occur.

\subsection{Partial fills, retries, and stopping}

Let the policy generate settled increments $(\Delta x_j,\Delta b_j)$ at times $s_j$. Define the debt payment from the newly available cash by
\begin{equation}
p_j=\min\{D_j,K_j+\Delta b_j\}.
\end{equation}
After each confirmation,
\begin{align}
q_{j+1}&=q_j-\Delta x_j,\\
D_{j+1}&=D_j-p_j,\\
K_{j+1}&=K_j+\Delta b_j-p_j.
\end{align}
Thus already applied cash is not counted again, debt never becomes negative, and any excess remains as confirmed surplus cash. The controller recomputes the remaining requirement using the latest debt and a refreshed uncertainty set. It never treats a match as cash.

The policy stops submitting new sale quantity when confirmed cash is sufficient to extinguish debt. Orders already matched may settle later and create overshoot. The implementation must record such surplus explicitly rather than manufacturing negative debt.

\subsection{Discrete-lot overshoot}

If the minimum executable lot is $\delta_q$ and the maximum net unit proceeds in the marginal lot are $\overline p$, then quantity overshoot is at most $\delta_q$ and cash overshoot is at most $\delta_q\overline p$, excluding already matched asynchronous batches. A separate bound for pending matched quantity is required when cancellation cannot prevent settlement.

\subsection{Illustrative robust decision}

The deterministic registry in \Cref{sec:deterministic} uses one entry position and three operating books. The base visible book would require fewer tokens than the worst operating book. The planned sale is therefore set by the lower envelope, while the realized audit minimum depends on which book path occurs. This deliberately creates conservative overshoot on favorable paths and exact coverage on the worst registered path.

\begin{figure}[H]
\centering
\includegraphics[width=0.88\textwidth]{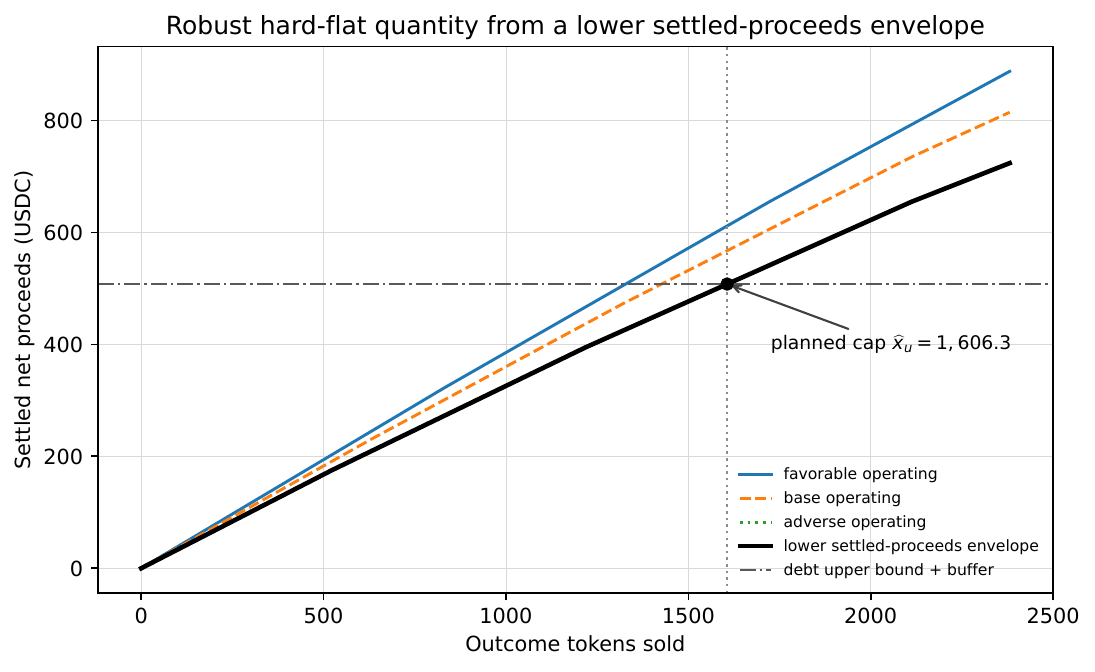}
\caption{Illustrative robust sale selection. Individual settled-proceeds paths differ; the lower envelope determines the ex-ante planned sale $\widehat x_u$. After settlement, the realized audit minimum $x_{\sigma}^{\ast}$ may be smaller.}
\label{fig:robust-proceeds}
\end{figure}

\subsection{Full-close policy}

A user may elect full close rather than residual spot conversion. Full close sells the entire executable quantity and repays debt first. It gives up \Cref{prop:residual-maximality} by choice but does not weaken lender safety. If full sale still cannot repay debt, the position enters the same failure and reserve path as auto-deleverage.

\section{Core Safety Results}
\label{sec:safety}

The principal results separate an ex-ante guarantee inside the registered operating set from pathwise accounting invariants after settlement.

\subsection{Robust ex-ante debt clearing}

\begin{theorem}[Robust ex-ante debt clearing]
\label{thm:robust-clearing}
Fix decision time $u$, horizon $\Delta$, policy $\Pi_u$, and operating set $\Uop_u(\Delta)$. Suppose:
\begin{enumerate}[label=(\roman*)]
  \item the planned sale $\widehat x_u$ in \eqref{eq:planned-sale} exists;
  \item the realized path $\omega$ belongs to $\Uop_u(\Delta)$;
  \item $T_c^{\mathrm{act}}(\omega)>u+\Delta$;
  \item \Cref{ass:authority,ass:lien,ass:settlement} hold throughout the horizon;
  \item the policy submits risk-reducing sales, applies every confirmed net proceed debt-first, and continues until either the debt is zero or cumulative settled quantity reaches $\widehat x_u$; every submitted fill used by the policy settles by $u+\Delta$; and
  \item no new borrowing is created after $u$.
\end{enumerate}
Then by some $\sigma\leq u+\Delta$,
\begin{equation}
D_{\sigma^+}=0
\qquad\text{and}\qquad
Z_{\sigma}\geq 0,
\label{eq:robust-clearing-result}
\end{equation}
where $Z_{\sigma}$ is confirmed dedicated cash remaining after debt repayment. If the policy is configured to continue until confirmed cash covers the realized debt plus $m_u$ (or if the entire certified cap $\widehat x_u$ settles), then at the corresponding completion time $\sigma_m\leq u+\Delta$,
\begin{equation}
D_{\sigma_m^+}=0
\qquad\text{and}\qquad
Z_{\sigma_m}\geq m_u.
\label{eq:robust-clearing-buffer-result}
\end{equation}
\end{theorem}

\begin{proof}
Because $\omega\in\Uop_u(\Delta)$,
\begin{equation*}
B_{u,\Delta}^{\mathrm{set},\Pi}(\widehat x_u,\omega)
\geq
\underline B_{u,\Delta}^{\Pi}(\widehat x_u).
\end{equation*}
Realized debt by the horizon is at most $\overline H_{u,\Delta}$. Hence, if cumulative settled quantity reaches $\widehat x_u$, confirmed dedicated cash is at least realized debt plus $m_u$. Under the stated policy, either debt reaches zero earlier or the certified cap settles; in both cases debt is extinguished by $u+\Delta$. Execution authority and collateral non-escape make the risk-reducing sale and debt-first application enforceable, while settlement observability prevents matched but unsettled amounts from being counted. If the policy stops at the first debt-zero prefix, only $Z_\sigma\geq0$ is guaranteed. If it continues to the registered debt-plus-buffer target, or the entire cap settles, the stronger bound $Z_{\sigma_m}\geq m_u$ follows.
\end{proof}

The theorem is conditional on the operating set. It does not state that the path remains inside the set or that a venue cannot close early. Those cases are treated in \Cref{sec:impossibility}.

\subsection{Pathwise debt-clearing invariant}

\begin{proposition}[Debt-clearing accounting invariant]
\label{prop:pathwise-clearing}
On any realized path, if confirmed dedicated cash and settled sale proceeds satisfy
\begin{equation}
K_u+B_{\sigma}^{\mathrm{audit}}(x^{\mathrm{sold}})\geq D_{\sigma^-}
\label{eq:pathwise-cover}
\end{equation}
and the application waterfall pays the loan before releasing residual assets, then
\begin{equation}
D_{\sigma^+}=0.
\label{eq:debt-zero}
\end{equation}
\end{proposition}

\begin{proof}
The waterfall applies at least $D_{\sigma^-}$ to a receivable that cannot be reduced below zero. Hence the post-application balance is zero.
\end{proof}

This proposition is intentionally labeled an invariant rather than an execution theorem. The substantive control result is \Cref{thm:robust-clearing}, which selects the sale before the path is known.

\begin{corollary}[No residual loan lien]
\label{cor:no-lien}
Under \Cref{prop:pathwise-clearing} and \Cref{ass:no-rehyp}, the extinguished loan has no claim on the residual outcome tokens after $\sigma$.
\end{corollary}

\subsection{Debt-free finality}

Let $q_{\mathrm{spot}}$ be the residual fully funded token quantity and $Z_{\sigma}\geq0$ the residual dedicated cash after repayment. At final payout time,
\begin{equation}
V_{\tau_f}^{\mathrm{claim}}
=Z_{\sigma}+q_{\mathrm{spot}}\pi_h.
\label{eq:claim-value-final}
\end{equation}

\begin{theorem}[Debt-free finality]
\label{thm:debt-free-finality}
Under \Cref{ass:collateralization,ass:no-rehyp} and confirmed $D_{\sigma^+}=0$, subsequent variation in the final payout vector creates no new credit-principal claim through the extinguished loan:
\begin{equation}
L_{\mathrm{loan}}^{\mathrm{credit}}(\boldsymbol\pi)=0
\qquad
\text{for every admissible }\boldsymbol\pi.
\label{eq:loan-loss-zero}
\end{equation}
The payout changes only the value of the user's residual claim, absent a separate operator guarantee.
\end{theorem}

\begin{proof}
The receivable is zero at $\sigma^+$. Redemption is paid by the outcome-token collateral system, and the residual token is not reattached to the loan. Therefore the payout vector cannot change the extinguished receivable.
\end{proof}

Whenever derivatives exist,
\begin{equation}
\frac{\partial L_{\mathrm{loan}}^{\mathrm{credit}}}{\partial\pi_h}=0,
\qquad
\frac{\partial V_{\tau_f}^{\mathrm{claim}}}{\partial\pi_h}=q_{\mathrm{spot}}.
\label{eq:payout-derivatives}
\end{equation}
The qualifier ``through the extinguished loan'' matters. If the same legal entity also guarantees custody or redemption, it can retain liabilities through those separate roles.

\subsection{Dispute-duration invariance}

Let $\Delta_f=\tau_f-\sigma$ be the interval from debt extinction to final payout. Let $\Delta_r=\tau_r-\tau_f$ be redemption latency.

\begin{theorem}[Loan-channel dispute-duration invariance]
\label{thm:dispute-invariance}
Suppose no new borrowing, delay compensation, or redemption guarantee is attached to the extinguished loan after $\sigma$. Then
\begin{equation}
\frac{\partial L_{\mathrm{loan}}^{\mathrm{credit}}}{\partial\Delta_f}=0,
\qquad
\frac{\partial L_{\mathrm{loan}}^{\mathrm{credit}}}{\partial\Delta_r}=0.
\label{eq:duration-invariance}
\end{equation}
The user may still bear opportunity cost, custody risk, legal uncertainty, and platform-access risk during both intervals.
\end{theorem}

\begin{proof}
Neither interval changes a receivable that was fully paid at $\sigma$. The conclusion does not apply to separate contractual obligations incurred by the operator.
\end{proof}

\begin{figure}[H]
\centering
\includegraphics[width=0.88\textwidth]{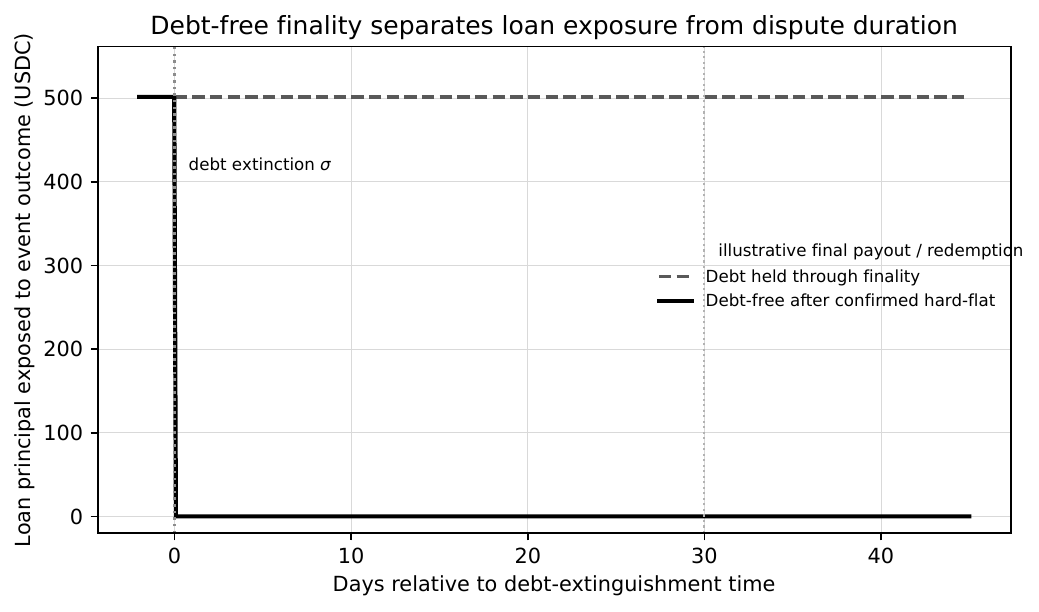}
\caption{Loan-channel exposure through finality. A financed claim held through resolution keeps principal exposed. Under confirmed debt-free conversion, the loan receivable reaches zero at $\sigma$; finality and redemption latency remain user and service-layer risks rather than event-outcome credit risk.}
\label{fig:dispute-invariance}
\end{figure}

\subsection{Positive debt cannot be outcome-invariant without support}

\begin{proposition}[Positive-debt impossibility at an adverse payout]
\label{prop:positive-debt-impossibility}
Suppose the held token admits final payout $\pi_h=0$. At finality, let dedicated cash be $K\geq0$, debt $D>0$, and independent external collateral or guarantee be zero. Then lender shortfall is
\begin{equation}
\ell_{\pi_h=0}=\pospart{D-K}.
\label{eq:adverse-payout-shortfall}
\end{equation}
If $K<D$, zero lender loss across all payout vectors is impossible while debt survives.
\end{proposition}

\begin{proof}
At $\pi_h=0$, the token pays nothing. Only dedicated cash remains available to the loan.
\end{proof}

\begin{corollary}[Necessary balance-sheet alternatives]
A mechanism seeking zero lender loss for every admissible payout must, before finality, either extinguish the debt, hold independent collateral covering the uncovered amount, obtain an enforceable guarantee, or exclude the zero-payout state by contract.
\end{corollary}

\subsection{State-machine invariant}

Define
\begin{equation}
\begin{aligned}
\mathcal G=\{&\state{DEBT\_FREE\_SPOT},\state{PENDING\_FINALITY},
\state{DELAYED},\state{DISPUTED},\\
&\state{FINAL},\state{REDEMPTION\_PENDING},\state{REDEEMED}\}.
\end{aligned}
\end{equation}

\begin{theorem}[No-debt finality-state invariant]
\label{thm:state-invariant}
If entry into $\mathcal G$ is guarded by atomic confirmed $D=0$, no transition inside $\mathcal G$ can create borrowing, and every transition uses the same ledger state as its guard, then every reachable state satisfies
\begin{equation}
S_t\in\mathcal G\Longrightarrow D_t=0.
\label{eq:state-invariant}
\end{equation}
\end{theorem}

\begin{proof}
Induction over state transitions. The entry guard establishes the base case; the no-borrow rule preserves it; atomicity eliminates a check-use race.
\end{proof}

Failure states such as \state{EVENT\_CLOSE\_EXCEPTION} and \state{CREDIT\_SHORTFALL} are deliberately excluded from $\mathcal G$ and may contain positive debt pending reserve or recovery.

\subsection{Ledger conservation}

\begin{proposition}[Balanced-ledger preservation]
\label{prop:ledger-invariant}
Under \Cref{ass:ledger}, if
\begin{equation}
\text{Assets}=\text{Liabilities}+\text{Equity}
\end{equation}
holds before a sequence of journal entries, it holds after each balanced entry in the sequence.
\end{proposition}

\begin{proof}
Each entry posts equal total debits and credits; induction preserves the identity. Idempotency prevents duplicate posting under the same event identifier.
\end{proof}

Balanced accounting does not prove authorization, correct valuation, or operator honesty. Those are security properties, not consequences of double-entry arithmetic.

\section{Exact Leverage, Margin, and Hard-Flat Timing}
\label{sec:risk}

The robust controller must constrain leverage before a position is opened. The exact condition is book- and policy-dependent; a scalar price ratio is only a benchmark.

\subsection{Exact book-dependent leverage envelope}

For candidate leverage $L\geq1$, confirmed entry quantity and residual cash are $q_0(L)$ and $K_0(L)$ from \eqref{eq:entry-quantity}--\eqref{eq:entry-cash}. Let $u(L)$ be the planned hard-flat decision time, $\Delta(L)$ the registered execution horizon, and $m(L)$ the hard-flat buffer. Define
\begin{equation}
\Psi(L)
=
K_0(L)
+
\underline B_{u(L),\Delta(L)}^{\Pi}\!\left(q_0(L)\right)
-
\overline H_{u(L),\Delta(L)}(L)
-
m(L).
\label{eq:exact-feasibility-margin}
\end{equation}

\begin{definition}[Exact admissible leverage set]
\begin{equation}
\mathcal L^{\mathrm{exact}}
=
\left\{L\in\mathcal L_{\mathrm{policy}}:\Psi(L)\geq0\right\},
\qquad
L_{\max}^{\mathrm{exact}}=\max\mathcal L^{\mathrm{exact}},
\label{eq:exact-leverage-envelope}
\end{equation}
where $\mathcal L_{\mathrm{policy}}$ is the finite set of leverage tiers permitted by product policy.
\end{definition}

The finite definition does not assume that $\Psi(L)$ is monotone. Fees, lot thresholds, changing execution policy, and concentration add-ons can create local non-monotonicity. An implementation should evaluate every permitted tier unless monotonicity has been proved for its exact parameterization.

\begin{proposition}[Exact debt-clearing leverage condition]
\label{prop:exact-leverage}
If $L\in\mathcal L^{\mathrm{exact}}$ and the assumptions of \Cref{thm:robust-clearing} hold at hard-flat, then the entire position quantity is sufficient to clear the debt and buffer inside the registered operating set.
\end{proposition}

\begin{proof}
$\Psi(L)\geq0$ is exactly the full-quantity robust feasibility inequality. Therefore the planned feasible set is non-empty, and \Cref{thm:robust-clearing} applies.
\end{proof}

\subsection{Linear-execution benchmark}

For interpretation only, suppose
\begin{equation}
A(q)=a q,
\qquad
\underline B(q)=b q,
\qquad
K_0=0,
\end{equation}
with $a>0$, $b\geq0$, debt-growth factor $\rho_D\geq1$, and zero additional buffer. Then $q=LC/a$, debt at the horizon is $\rho_D(L-1)C$, and the clearance condition becomes
\begin{equation}
\frac{b}{a}L\geq\rho_D(L-1).
\end{equation}
Let $g=b/a$.

\begin{corollary}[Linear recovery-ratio benchmark]
\label{cor:linear-envelope}
If $g<\rho_D$, then
\begin{equation}
L\leq
L_{\max}^{\mathrm{lin}}
=
\frac{\rho_D}{\rho_D-g}.
\label{eq:linear-leverage-bound}
\end{equation}
If $g\geq\rho_D$, the simplified inequality does not bind and external policy caps determine leverage.
\end{corollary}

The formula is not used for production sizing. Real $a$ and $b$ depend on quantity, shared depth, fees, partial fills, and time. The exact envelope in \eqref{eq:exact-leverage-envelope} is primary.

\subsection{Comparative statics}

In the linear benchmark, $L_{\max}^{\mathrm{lin}}$ increases in $g$ and decreases in $\rho_D$. In the exact model, the same directional intuition generally holds but is not guaranteed across discrete tiers. A deterioration in the lower proceeds envelope, increase in settlement horizon, or increase in buffer can only remove a tier if all other objects are fixed.

\subsection{Maintenance liquidation value}

At monitoring time $t<u$, define a shorter-horizon lower proceeds envelope $\underline V_t(q)$ under the liquidation policy. Maintenance health is
\begin{equation}
\hf_t(q)=\frac{K_t+\underline V_t(q)}{D_t+M_t(q)}.
\end{equation}
The position is liquidatable when $\hf_t\leq1$. A target $\overline H>1$ may be used after partial liquidation.

For a candidate partial sale $y$, let $\underline B_t^{\mathrm{part}}(y)$ be a lower settled-proceeds bound for that sale. Let
\begin{equation}
\underline V_t^{\mathrm{post}}(q-y\mid y)
=\inf_{\omega\in\mathcal U_t^{\mathrm{liq}}}
V_t^{\mathrm{post}}(q-y\mid y,\omega)
\label{eq:post-sale-residual-value}
\end{equation}
be the lower liquidation value of the residual quantity after the same scenario path has consumed the depth used by $y$ and applied any registered book response. It is not a fresh top-of-book valuation. Define the amount of confirmed cash applied to debt by
\begin{equation}
P_t(y)=\min\{D_t,K_t+\underline B_t^{\mathrm{part}}(y)\},
\end{equation}
and the conservative post-settlement state
\begin{align}
D_t'(y)&=D_t-P_t(y),\\
K_t'(y)&=K_t+\underline B_t^{\mathrm{part}}(y)-P_t(y),\\
q_t'(y)&=q-y.
\end{align}
A partial liquidation sells the minimum admissible $y$ such that
\begin{equation}
\frac{K_t'(y)+\underline V_t^{\mathrm{post}}(q_t'(y)\mid y)}
{D_t'(y)+M_t(q_t'(y))}
\geq \overline H.
\label{eq:partial-liquidation-condition}
\end{equation}
Using separate lower infima for sale cash and post-sale residual value is conservative even when their worst paths differ; conditioning the residual value on consumed depth prevents internal liquidity reuse. The actual debt application and settlement sequence is spelled out in \Cref{app:algorithms}. If the buffer has discontinuous tiers, exhaustive lot-grid evaluation is safer than assuming monotonicity.

\subsection{Depth, staleness, and confidence gates}

New leverage is prohibited when any of the following holds:
\begin{enumerate}[label=(G\arabic*)]
  \item no synchronized book or RFQ quote exists;
  \item order-book sequence or hash reconciliation fails;
  \item the snapshot age exceeds the registered threshold;
  \item the full-position lower proceeds envelope is below the debt-and-buffer threshold;
  \item execution authority or collateral control is unavailable;
  \item actual market state is not open;
  \item aggregate risk-bucket capacity in \Cref{sec:aggregate} is exhausted; or
  \item the remaining interval before scheduled close is below the hard-flat horizon plus safety buffer.
\end{enumerate}

The confidence score is a gating variable, not a substitute for proceeds. A high-confidence midpoint does not repay a loan.

\subsection{Scheduled and actual close}

Let
\begin{equation}
\Delta_{\mathrm{HF}}
=
\delta_{\mathrm{quote}}
+
\delta_{\mathrm{submit}}
+
\delta_{\mathrm{match}}
+
\delta_{\mathrm{retry}}
+
\delta_{\mathrm{settle}}
+
\delta_{\mathrm{clock}}
+
\delta_{\mathrm{early}},
\label{eq:hard-flat-buffer}
\end{equation}
where $\delta_{\mathrm{early}}$ is a policy allowance for early close or suspension risk. Scheduled hard-flat must satisfy
\begin{equation}
u\leq T_c^{\mathrm{sched}}-\Delta_{\mathrm{HF}}.
\label{eq:scheduled-hard-flat}
\end{equation}
This does not guarantee $u<T_c^{\mathrm{act}}$. If the venue can close administratively without advance notice, that event remains in $\Ufail$ unless covered by an enforceable close-window agreement.

\subsection{Time-to-close leverage compression}

A simple product cap is
\begin{equation}
L_{\max}(t)
=
\min\left\{
L_{\max}^{\mathrm{exact}}(t),
1+(L_{\mathrm{far}}-1)\chi(t)
\right\},
\label{eq:time-compression}
\end{equation}
where $\chi(t)$ decreases from one to zero between reduce-only and hard-flat. At $u$, no new leverage is allowed. This schedule reduces the volume that must be sold simultaneously but does not replace the exact feasibility test.

\subsection{Hard-flat shortfall surface}

For the linear benchmark, shortfall per unit user collateral is
\begin{equation}
\ell(L,g)
=
\pospart{\rho_D(L-1)-gL}.
\label{eq:linear-shortfall}
\end{equation}
The surface is useful for intuition and verifier tests. It is not a substitute for the nonlinear shared-book model.

\begin{figure}[H]
\centering
\includegraphics[width=0.88\textwidth]{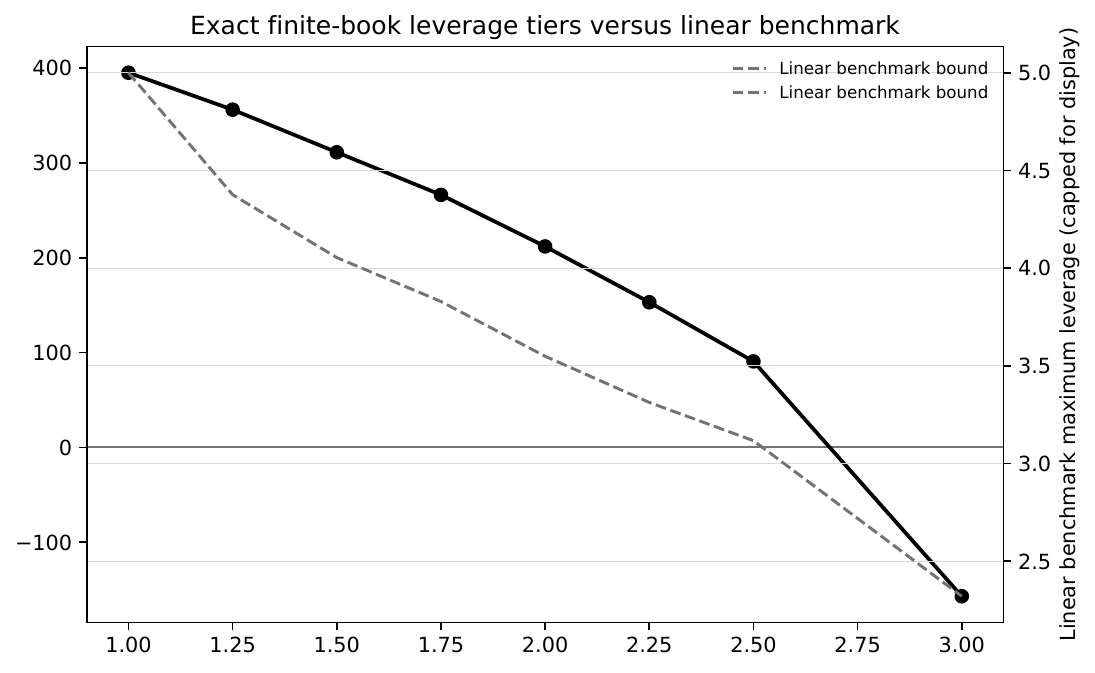}
\caption{Illustrative exact leverage tiers under a finite entry ask ladder and robust exit envelope. The scalar linear benchmark is shown only for comparison; admissibility follows the finite book-dependent feasibility margin.}
\label{fig:exact-leverage}
\end{figure}

\section{Aggregate Hard-Flat and Shared Liquidity}
\label{sec:aggregate}

Position-level feasibility is insufficient when several loans rely on the same bid book. If every position independently values itself from the top of the same ladder, the system reuses liquidity that can be consumed only once.

\subsection{Shared-book proceeds}

Consider $n$ positions in the same token and hard-flat window. Position $i$ has quantity $q_i$, dedicated cash $K_i$, and debt-service upper bound $\overline H_i$. Let
\begin{equation}
X=\sum_{i=1}^{n}x_i
\end{equation}
be aggregate sale quantity. Let $\mathcal B_{u,\Delta}^{\mathrm{set},\Pi}(X,\omega)$ be pathwise settled proceeds from executing aggregate quantity $X$ through the shared venue book and policy. Define
\begin{equation}
\underline{\mathcal B}_{u,\Delta}(X)
=
\inf_{\omega\in\mathcal U_{u,\mathrm{agg}}^{\mathrm{op}}(\Delta)}
\mathcal B_{u,\Delta}^{\mathrm{set},\Pi}(X,\omega).
\label{eq:aggregate-proceeds}
\end{equation}
This is not the sum of $n$ independently re-anchored curves.

\begin{proposition}[Top-of-book double counting]
\label{prop:double-counting}
Suppose a common cumulative proceeds curve $B$ has non-increasing marginal proceeds and $B(0)=0$. If each position is valued independently from the same initial book, then
\begin{equation}
\sum_{i=1}^{n}B(q_i)
\geq
B\!\left(\sum_{i=1}^{n}q_i\right),
\label{eq:double-counting}
\end{equation}
with strict inequality whenever the aggregate sale walks into worse marginal levels not reached by every position individually.
\end{proposition}

\begin{proof}
A cumulative proceeds curve with non-increasing marginal proceeds is subadditive on non-negative quantities: the marginal units in $B(x+y)$ are no better than the units already counted in $B(x)$ and a separately re-anchored $B(y)$. Repeating the argument gives the result.
\end{proof}

The proposition formalizes a system-level failure mode: every loan can appear individually covered while the combined liquidation is not.

\subsection{Pooled aggregate feasibility}

If the legal and accounting structure permits proceeds within a risk bucket to repay the bucket's loans as a pool, define
\begin{equation}
\widehat{\mathcal F}^{\mathrm{pool}}_{u,\Delta}
=
\left\{
\mathbf x:
\sum_i K_i+
\underline{\mathcal B}_{u,\Delta}\!\left(\sum_i x_i\right)
\geq
\sum_i\overline H_i+M_u^{\mathrm{agg}}
\right\},
\label{eq:pooled-feasible-set}
\end{equation}
where $M_u^{\mathrm{agg}}$ is a shared operational and concentration buffer.

\begin{theorem}[Aggregate pooled debt clearing]
\label{thm:aggregate-pooled}
If $\mathbf x\in\widehat{\mathcal F}^{\mathrm{pool}}_{u,\Delta}$, the realized path lies in the registered aggregate operating set, the shared execution settles by the horizon, and the bucket waterfall may legally pool cash across the included loans, then all bucket debt can be extinguished and at least $M_u^{\mathrm{agg}}$ remains.
\end{theorem}

\begin{proof}
The aggregate lower envelope and debt-service upper bounds imply that confirmed bucket cash is at least total bucket debt plus the buffer. Pooling authority permits application to every included receivable.
\end{proof}

This theorem gives system solvency, not a unique allocation of residual user exposure.

\subsection{Segregated loans and deterministic execution priority}

If each loan must be repaid only from its own sold tokens and cash, the mechanism needs a deterministic priority order $\pi=(\pi_1,\ldots,\pi_n)$. Let
\begin{equation}
X_{j-1}=\sum_{k<j}x_{\pi_k}.
\end{equation}
The incremental lower proceeds allocated to position $\pi_j$ are defined directly from pathwise increments:
\begin{equation}
\underline B^{\pi_j\mid\pi_{<j}}_{u,\Delta}(x_{\pi_j})
=
\inf_{\omega\in\mathcal U_{u,\mathrm{agg}}^{\mathrm{op}}(\Delta)}
\left[
\mathcal B_{u,\Delta}^{\mathrm{set},\Pi}(X_{j-1}+x_{\pi_j},\omega)
-
\mathcal B_{u,\Delta}^{\mathrm{set},\Pi}(X_{j-1},\omega)
\right].
\label{eq:incremental-proceeds}
\end{equation}
Taking a difference of two independently minimized aggregate envelopes would not in general be a valid lower bound, because the minimizing paths can differ. Equation~\eqref{eq:incremental-proceeds} avoids that error.

\begin{definition}[Priority-feasible allocation]
An allocation $(\pi,\mathbf x)$ is priority-feasible if, for every $j$,
\begin{equation}
K_{\pi_j}
+
\underline B^{\pi_j\mid\pi_{<j}}_{u,\Delta}(x_{\pi_j})
\geq
\overline H_{\pi_j}+m_{\pi_j}.
\label{eq:priority-feasibility}
\end{equation}
\end{definition}

\begin{theorem}[Segregated aggregate clearing without depth reuse]
\label{thm:aggregate-segregated}
Under the operating-set, authority, and settlement assumptions of \Cref{thm:robust-clearing}, a priority-feasible allocation clears every loan in order $\pi$ without assigning the same unit of shared-book proceeds to more than one position.
\end{theorem}

\begin{proof}
For every realized path, pathwise increments telescope to the pathwise aggregate proceeds, so each marginal unit is assigned once. Each certified incremental lower bound is no larger than the realized increment on that path. Hence every position's inequality covers its debt and buffer after accounting for the quantity consumed by earlier positions. Apply \Cref{thm:robust-clearing} sequentially.
\end{proof}

\subsection{Canonical priority policy and fairness}

For priority selection only, define position $i$'s standalone robust coverage ratio
\begin{equation}
\Gamma_i(u)
=
\frac{K_i+\underline B^{\mathrm{standalone}}_{i,u,\Delta}(q_i)}
{\overline H_i+m_i},
\label{eq:robust-coverage-ratio}
\end{equation}
with $\Gamma_i=+\infty$ when the denominator is zero. The standalone numerator is evaluated on the same registered scenario version but is not used as an aggregate solvency value; shared-book feasibility continues to use \eqref{eq:incremental-proceeds}. A lower $\Gamma_i$ identifies a position closer to lender-principal shortfall under its own full liquidation.

The canonical emergency order is \emph{lowest robust coverage first}: ascending $\Gamma_i$, then earliest hard-flat deadline, then deterministic position identifier. The ordering is frozen in the aggregate certificate before orders are submitted. It avoids rewarding API latency or operator discretion and gives scarce best marginal depth to the positions closest to principal loss. No priority rule is welfare-neutral: later positions may retain less residual exposure or receive worse execution.

\begin{proposition}[Priority determinism]
\label{prop:priority-determinism}
For a finite position set with finite $\Gamma_i$, totally ordered deadline timestamps, and unique deterministic identifiers, the canonical lexicographic rule produces a unique execution order. If exact $\Gamma_i$ and deadline ties are grouped for pro-rata processing, the ordered sequence of tie cohorts is unique and any indivisible-lot residual is uniquely assigned by identifier.
\end{proposition}

\begin{proof}
Lexicographic order on the tuple $(\Gamma_i,\text{deadline}_i,\text{identifier}_i)$ is a total order once identifiers are unique. Grouping the first two equal components produces a unique ordered partition, and the identifier order resolves any discrete residual within a cohort.
\end{proof}

In the normal operating region, admission control requires aggregate feasibility before leverage is issued. If the registered shared-book inequalities hold, every admitted loan clears and priority affects residual exposure rather than solvency. If the realized path leaves $\Uop$, the same frozen order governs recovery. Exact ties may be processed as one cohort with sale targets proportional to debt-service requirement; indivisible-lot residuals are assigned by position identifier. Reserve capital, if needed after aggregate execution, follows the pro-rata principal-shortfall rule in \Cref{sec:impossibility}.

\subsection{Aggregate open-interest capacity}

For a risk bucket with position vector $\mathbf q$, define the full-liquidation robust coverage margin
\begin{equation}
\Psi_{\mathrm{agg}}(\mathbf q)
=
\sum_i K_i+
\underline{\mathcal B}_{u,\Delta}\!\left(\sum_i q_i\right)
-
\sum_i\overline H_i
-
M_u^{\mathrm{agg}}.
\label{eq:aggregate-margin}
\end{equation}
New leverage is admitted only if the post-trade vector keeps $\Psi_{\mathrm{agg}}\geq0$ under the registered close scenario and the selected priority rule remains feasible where loan segregation requires it.

A scalar OI cap is a lossy summary:
\begin{equation}
\mathrm{OI}_{\max}
=
\sup\left\{
\mathrm{OI}:\Psi_{\mathrm{agg}}(\mathbf q(\mathrm{OI}))\geq0
\right\}.
\label{eq:aggregate-oi-cap}
\end{equation}
The cap must be recomputed when depth, fees, settlement horizon, or concentration changes.

\subsection{Cross-market clusters}

For several tokens or correlated markets, replace the scalar aggregate curve with
\begin{equation}
\underline{\mathcal B}_{u,\Delta}(\mathbf x)
\end{equation}
that captures joint execution under a common news shock, venue suspension, or shared market maker. Summing independent single-market expected losses is not sufficient. The scenario set must include simultaneous book withdrawal and common settlement delays.

\begin{figure}[H]
\centering
\includegraphics[width=0.88\textwidth]{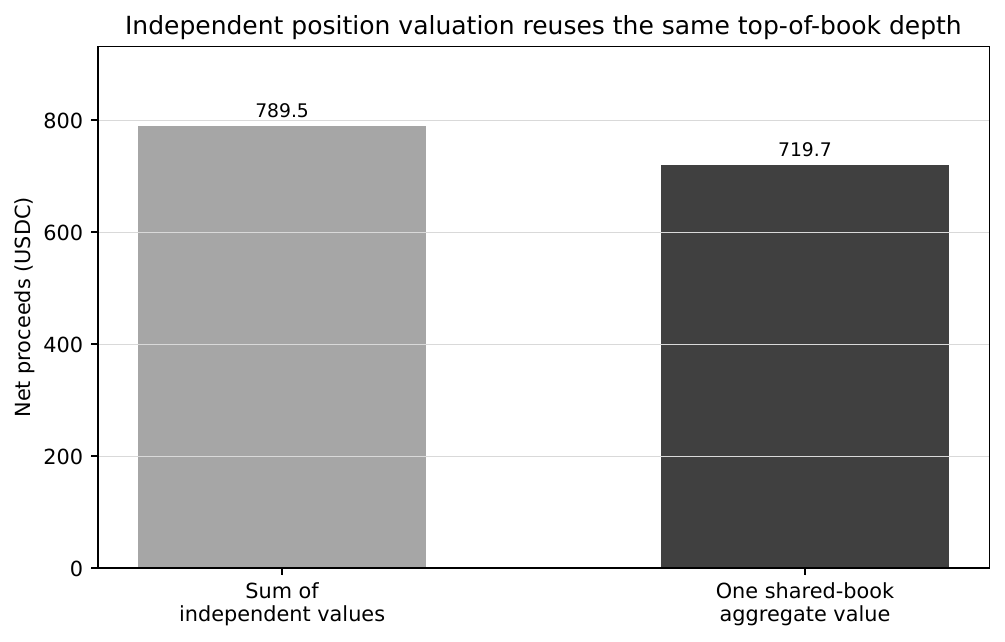}
\caption{Shared-book liquidity cannot be reused. The sum of three independently re-anchored position values exceeds the proceeds from selling the combined quantity through one finite bid ladder. Aggregate capacity must be computed on the joint flow.}
\label{fig:aggregate-liquidity}
\end{figure}

\subsection{Implication for reserves}

Reserve calculations in \Cref{sec:impossibility} use aggregate cluster proceeds. Position-by-position shortfalls may be reported for attribution, but the total reserve requirement must be based on one shared execution path per scenario, not on several copies of the same book.

\section{Resolution, Dispute, and Redemption}
\label{sec:finality}

Debt-free finality does not eliminate oracle or redemption states. It prevents those states from carrying the ordinary loan.

\subsection{Resolution state space}

A venue-specific lifecycle may contain more states, but the canonical abstraction is
\begin{equation}
\state{OPEN}
\rightarrow
\state{CLOSED}
\rightarrow
\state{PENDING\_RESOLUTION}
\rightarrow
\begin{cases}
\state{FINAL},\\
\state{DELAYED},\\
\state{DISPUTED},\\
\state{VOIDED},\\
\state{CANCELLED}.
\end{cases}
\label{eq:resolution-state-space}
\end{equation}
A delayed or disputed market may later return to pending resolution, become final with an override, or be voided.

\begin{definition}[Final payout event]
The payout vector is final at $\tau_f$ when the outcome-token system accepts a vector $\boldsymbol\pi$ for redemption and ordinary challenge or review no longer changes that vector.
\end{definition}

A venue API status is not by itself finality if the on-chain or contractual payout system has not accepted the vector.

\subsection{A proposal is not final accounting}

Let $\widehat{\boldsymbol\pi}_{\tau_p}$ be a provisional proposal. It may be confirmed, overridden, or voided.

\begin{proposition}[No provisional settlement]
\label{prop:no-provisional}
Before $\tau_f$, the mechanism must not use $\widehat{\boldsymbol\pi}_{\tau_p}$ to post final user wealth, distribute final redemption cash, or retrospectively reopen an extinguished loan.
\end{proposition}

\begin{proof}
By definition, the proposal is not the authoritative payout vector and may change. Treating it as final can create transfers inconsistent with the actual collateral redemption. The extinguished loan is already settled at $\sigma$ and has no basis for reopening.
\end{proof}

The proposal may be displayed as provisional information and used by a separate oracle-monitoring process.

\subsection{Disputed and delayed states}

When the market is disputed or delayed:
\begin{enumerate}[label=(\roman*)]
  \item ordinary financed positions must already have $D=0$ or be in an explicit exception state;
  \item no final PnL is posted;
  \item no new interest accrues on the extinguished loan;
  \item historical last trade is not presented as redeemable value;
  \item the residual token remains a claim against the outcome collateral system; and
  \item the user interface distinguishes claim value, provisional proposal, and available cash.
\end{enumerate}

A long dispute changes user opportunity cost but, under \Cref{thm:dispute-invariance}, does not change credit-principal loss through the extinguished loan.

\subsection{Void and cancellation}

A void is modeled through the actual payout vector, not through cost-basis reimbursement. For held token $h$, final claim value is always
\begin{equation}
q_{\mathrm{spot}}\pi_h.
\end{equation}
If a venue uses equal payouts, the formula produces that result. If it uses another vector, the mechanism follows the vector. A cancellation that returns collateral through a separate process is represented by the economically equivalent final vector and redemption event.

\subsection{Claim finality and cash availability}

At $\tau_f$, the user has a final claim value
\begin{equation}
V_{\tau_f}^{\mathrm{claim}}=Z_{\sigma}+q_{\mathrm{spot}}\pi_h.
\end{equation}
Cash is available only after confirmed redemption at $\tau_r$:
\begin{equation}
W_{\tau_r}^{\mathrm{cash}}
=Z_{\sigma}+\operatorname{Redeem}_{\tau_r}(q_{\mathrm{spot}},\boldsymbol\pi)-c_{\mathrm{red}},
\label{eq:cash-after-redemption}
\end{equation}
where $c_{\mathrm{red}}$ is the confirmed redemption cost if any.

This distinction answers a common ambiguity. Asynchronous redemption does not invalidate debt-free finality; it delays conversion of the user's final claim into spendable cash. If the operator guarantees redemption timing or amount, that guarantee is a separate liability.

\subsection{Oracle watcher and challenge policy}

The reference implementation treats oracle monitoring as a separate service. It:
\begin{enumerate}[label=(O\arabic*)]
  \item records the resolution text and authoritative sources before listing;
  \item observes proposals and payout vectors;
  \item compares them with the registered criteria;
  \item opens an incident when the proposal appears inconsistent;
  \item prepares evidence and challenge instructions; and
  \item records the final on-chain or contractual outcome.
\end{enumerate}

The base mechanism does not allow an unreviewed model or end user to spend shared challenge capital automatically. Challenge bonds are segregated from the lender reserve because they support oracle governance rather than loan losses.

\subsection{Dispute-safe user disclosure}

Before entry, the user should see:
\begin{itemize}
  \item scheduled reduce-only and hard-flat times;
  \item that actual close can occur earlier;
  \item whether residual spot is the default or full close is selected;
  \item that dispute can lock the residual claim for an unknown period;
  \item that provisional results are not final payouts;
  \item that void payment follows the venue payout vector, not purchase price; and
  \item that redemption may remain pending after finality.
\end{itemize}

These disclosures are part of mechanism correctness because they define the economic object the user is buying. They are not a substitute for technical enforcement.

\section{Failure Boundaries and Reserve Sufficiency}
\label{sec:impossibility}

The operating set supports a conditional guarantee. This section states what happens outside it and corrects reserve calculations for shared liquidity.

\subsection{Operating set versus failure set}

At decision time $u$, $\Uop_u(\Delta)$ contains paths for which the controller claims the robust certificate. A typical operating set may bound:
\begin{itemize}
  \item cumulative bid cancellation and price deterioration;
  \item selected $T_{k,h,s}$ book transformations;
  \item per-attempt partial fill;
  \item fees and gas;
  \item retry count;
  \item settlement latency up to $\Delta$; and
  \item ordinary chain confirmation behavior.
\end{itemize}

The registered stress set $\Ustress_u(\Delta)$ contains $\Uop_u(\Delta)$ and a broader finite catalogue of adverse transformations, correlated hard-flat flow, control degradation, and settlement exceptions. Inclusion in $\Ustress$ means that a path is tested and reported; it does not imply that the robust certificate applies. $\Ufail_u(\Delta)=\Omega\setminus\Uop_u(\Delta)$ contains at least:
\begin{itemize}
  \item $T_c^{\mathrm{act}}\leq u+\Delta$;
  \item no executable bid quantity;
  \item loss or revocation of liquidation authority;
  \item borrower removal of collateral;
  \item persistent settlement failure;
  \item chain or contract failure beyond the confirmation assumption; and
  \item aggregate liquidation flow exceeding the registered shared-book capacity.
\end{itemize}

A path can move from operating to failure status after the decision. The controller must then stop presenting the robust certificate and enter recovery.

\subsection{No universal backend-only guarantee}

\begin{theorem}[No universal backend-only zero-shortfall guarantee]
\label{thm:no-universal}
Consider a mechanism with leverage $L>1$ that relies only on user collateral, the financed outcome tokens, and backend instructions to an external venue. Suppose no independent collateral, enforceable liquidity put, or third-party guarantee exists, and the mechanism admits a post-entry state with positive uncovered debt $D_t>K_t$. If admissible paths include any of:
\begin{enumerate}[label=(\alph*)]
  \item market closure before a debt-clearing sale settles;
  \item zero settled proceeds for every sale quantity while debt is positive;
  \item loss of enforceable authority to sell the pledged token; or
  \item removal of pledged collateral before debt application,
\end{enumerate}
then zero lender shortfall cannot be guaranteed over all admissible payout vectors and paths.
\end{theorem}

\begin{proof}
Because the admitted state has $D_t>K_t$, the loan is not already cash-covered. In each listed path, the controller can be prevented from converting the financed token into sufficient dedicated cash. Since the held token may pay zero, \Cref{prop:positive-debt-impossibility} produces a positive shortfall on such a path.
\end{proof}

\begin{corollary}[Additional structure required]
A universal pathwise guarantee with $L>1$ requires at least one of: full independent collateral, an enforceable liquidity or settlement guarantee, a third-party credit guarantee, a venue-recognized lien and liquidation right plus an admissible-path restriction, or exclusion of the failure paths by contract.
\end{corollary}

The theorem does not say the product cannot operate. It says the product must distinguish a certified operating region from externally capitalized tail states.

\subsection{Unexpected close shortfall}

If actual close occurs at $T_c^{\mathrm{act}}<\sigma$ and no further sale is possible, recovery at finality is at most dedicated cash plus token payout. Shortfall is
\begin{equation}
\ell^{\mathrm{close}}
=
\pospart{D_{T_c^{\mathrm{act}}}-K_{T_c^{\mathrm{act}}}-q\pi_h-G},
\label{eq:unexpected-close-shortfall}
\end{equation}
where $G$ is independent guaranteed recovery. In the adverse payout state with $G=0$, the uncovered debt is at risk.

\subsection{Scenario shortfall with shared execution}

Let $\mathcal C_k$ be a risk cluster sharing a token, venue, close window, oracle, chain, or liquidity provider. For scenario $s$, let $\mathcal B_k^{(s)}(\mathbf q_k)$ be aggregate recoverable settled proceeds for the whole cluster after the scenario's execution and settlement assumptions. Define
\begin{equation}
\ell_k(s)
=
\pospart{
\sum_{i\in\mathcal C_k}D_i(s)
-
\sum_{i\in\mathcal C_k}K_i(s)
-
\mathcal B_k^{(s)}(\mathbf q_k)
-
G_k(s)
}.
\label{eq:cluster-shortfall}
\end{equation}

This replaces the sum of independently re-anchored position proceeds. Position-level loss attribution may be computed after the aggregate execution using the registered priority rule.

\subsection{Finite-scenario reserve sufficiency}

Let $\Sset$ be a finite registered set of scenarios and let $\mathcal R$ be a segregated reserve legally available for the modeled credit losses.

\begin{theorem}[Aggregate scenario-set reserve sufficiency]
\label{thm:reserve-sufficiency}
The reserve covers every modeled aggregate shortfall in $\Sset$ if
\begin{equation}
\mathcal R
\geq
\sup_{s\in\Sset}
\sum_k\ell_k(s).
\label{eq:reserve-bound}
\end{equation}
\end{theorem}

\begin{proof}
For every registered scenario, the right-hand side dominates the total cluster shortfall. Availability of the segregated reserve permits payment of that amount.
\end{proof}

The theorem says nothing about scenarios omitted from $\Sset$, wrong execution models, unavailable reserve assets, or legal subordination.

\subsection{Minimum scenario catalogue}

The registered catalogue should include:
\begin{enumerate}[label=(S\arabic*)]
  \item ordinary hard-flat;
  \item spread widening and depth cancellation;
  \item adverse price movement before first fill;
  \item partial fill followed by worse-book retry;
  \item settlement delay with maximum debt accrual;
  \item matched trade followed by settlement failure;
  \item actual close before scheduled close;
  \item complete loss of executable liquidity;
  \item signer or authorization failure;
  \item borrower collateral-withdrawal attempt;
  \item simultaneous hard-flat across the largest shared book; and
  \item correlated venue, oracle, RPC, or chain impairment.
\end{enumerate}

The catalogue must be frozen before position approval and versioned. Adding a scenario after an incident is learning, not retrospective coverage.

\subsection{Reserve waterfall and allocation}

A transparent waterfall is
\begin{equation}
\begin{aligned}
\text{dedicated cash and collateral}
&\rightarrow \text{settled sale proceeds}
\rightarrow \text{position buffer}\\
&\rightarrow \text{risk-bucket reserve}
\rightarrow \text{global reserve}
\rightarrow \text{operator equity}.
\end{aligned}
\label{eq:waterfall}
\end{equation}
Socialized loss to unrelated profitable users is excluded from the base mechanism. If introduced later, it becomes an explicit contingent liability.

After shared-book execution and position attribution, let $\ell_i\geq0$ be the residual lender-principal shortfall of position $i$ at one reserve stage and let $R\geq0$ be the capital available at that stage. The canonical allocation is pro rata to principal shortfall:
\begin{equation}
r_i(R,\boldsymbol\ell)
=
\begin{cases}
\ell_i, & \sum_j\ell_j\leq R,\\[3pt]
R\,\dfrac{\ell_i}{\sum_j\ell_j}, & \sum_j\ell_j>R,
\end{cases}
\label{eq:pro-rata-reserve}
\end{equation}
with $r_i=0$ when $\sum_j\ell_j=0$. The same rule is applied separately at the risk-bucket and global-reserve stages; later stages receive only the residual shortfall.

\begin{proposition}[Reserve-allocation consistency]
\label{prop:reserve-allocation}
For a finite non-negative shortfall vector and reserve $R\geq0$, allocation \eqref{eq:pro-rata-reserve} is unique, satisfies $0\leq r_i\leq\ell_i$ for every $i$, and obeys
\begin{equation}
\sum_i r_i=\min\left\{R,\sum_i\ell_i\right\}.
\end{equation}
If the reserve is insufficient, every positive shortfall receives the same coverage fraction $R/\sum_j\ell_j$.
\end{proposition}

\begin{proof}
The fully covered branch is immediate. In the insufficient branch, non-negativity follows from $R,\ell_i\geq0$; because $R<\sum_j\ell_j$, each allocation is below $\ell_i$; and summing the proportional terms gives $R$. The formula fixes every component uniquely.
\end{proof}

\subsection{Reserve-supported open interest}

An approximate scalar rule for bucket OI is
\begin{equation}
\mathrm{OI}_{\max}^{\mathrm{reserve}}
=
\frac{\mathcal R}{c_R\lambda},
\label{eq:reserve-oi}
\end{equation}
where $\lambda$ is a conservative cluster stress-loss rate and $c_R\geq1$ a coverage multiple. The exact aggregate book constraint in \eqref{eq:aggregate-oi-cap} remains primary; reserve capacity cannot justify admitting more ordinary operating-set exposure than can be deleveraged.

\subsection{Capital cost and pricing}

Expected contribution can be written as
\begin{equation}
\mathrm{CM}
=
F_{\mathrm{protocol}}
+(r_{\mathrm{borrow}}-r_{\mathrm{capital}})D\Delta t
+R_{\mathrm{partner}}
-F_{\mathrm{venue}}
-\E[\ell]
-C_{\mathrm{ops}}
-C_{\mathrm{reserve}}.
\label{eq:contribution-margin}
\end{equation}
A positive borrow spread does not imply a viable product if venue fees, reserve capital, and correlated shortfall dominate.

\subsection{Why a reserve is not a proof}

Three quantities must be reported separately:
\begin{enumerate}[label=(\roman*)]
  \item robust debt-clearing feasibility before reserve;
  \item gross shortfall after dedicated position assets and guarantees; and
  \item residual loss after the reserve waterfall.
\end{enumerate}
A single ``covered'' label would hide whether safety came from the mechanism or from external capital after mechanism failure.

\begin{figure}[H]
\centering
\includegraphics[width=0.88\textwidth]{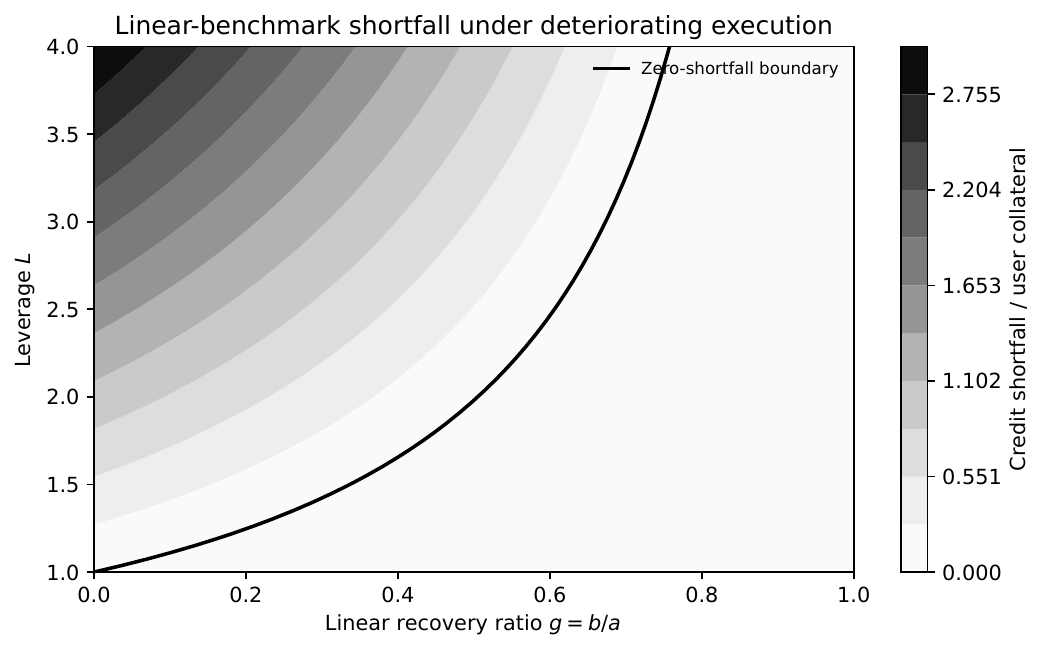}
\caption{Linear-benchmark hard-flat shortfall as leverage rises and executable recovery deteriorates. The figure is an intuition surface; actual reserve sizing uses nonlinear shared-book scenario proceeds.}
\label{fig:shortfall-surface}
\end{figure}

\section{Deterministic Verification Without External Data}
\label{sec:deterministic}

This revision does not estimate execution distributions. It verifies the formal mechanism on a registered deterministic scenario set that exposes the distinction between quoted, matched, settled, and aggregate proceeds.

\subsection{Verification objectives}

The accompanying script checks:
\begin{enumerate}[label=(V\arabic*)]
  \item exact finite-book acquisition and sale arithmetic;
  \item construction of the lower settled-proceeds envelope;
  \item planned robust sale before execution;
  \item realized audit minimum after settlement;
  \item debt-service upper bounds under settlement delay;
  \item conservative overshoot on favorable paths;
  \item debt-free payout and dispute invariance;
  \item detection of thin-depth and zero-liquidity infeasibility;
  \item exact leverage-tier feasibility versus the linear benchmark;
  \item nested operating and stress-set classification;
  \item adversarial $T_{k,h,s}$ book transformations;
  \item aggregate shared-book execution without top-of-book reuse;
  \item lowest-$\Gamma_i$ priority ordering; and
  \item pro-rata reserve allocation in shortfall scenarios.
\end{enumerate}

The suite is algebraic verification, not empirical evidence that any venue will remain inside the registered operating set.

\subsection{Registered position}

The base position uses
\begin{align}
C&=500, & L&=2, & N&=1000,\\
a_{\mathrm{eff}}&=0.42, & q_0&=2380.95238095,\\
r&=15\%\text{ per year}, & D_u&=500e^{0.15\cdot7/365}.
\end{align}
The ex-ante operating set includes three settled-proceeds paths: favorable, base, and adverse-but-operating. The broader stress registry includes those paths, deterministic $T_{k,h,s}$ transformations, a thin-depth path, zero liquidity, and control-capability failures. Only the three operating paths support the robust certificate.

Settlement horizons of 5, 30, and 120 seconds are represented through debt-service upper envelopes and scenario-specific charges. The numbers are illustrative and are not attributed to PredictStreet or any other venue.

\subsection{Robust sale result}

Let $\underline B$ be the pointwise minimum of the three operating settled-proceeds curves. The verifier computes
\begin{equation}
\widehat x_u
=
\min\left\{x:
\underline B(x)\geq\overline H_{u,\Delta}+m_u
\right\}.
\end{equation}
For every operating path, sale of at most $\widehat x_u$ clears the maximum registered debt and buffer. On favorable paths the audit minimum is smaller, and the difference is reported as conservative quantity overshoot.

\begin{center}
\small
\begin{tabular}{@{}lrrrr@{}}
\toprule
Path & Audit minimum & Planned cap & Residual after audit & Status \\
\midrule
Favorable operating & 1,311.2140 & 1,606.2949 & 1,069.7384 & debt cleared \\
Base operating & 1,409.2884 & 1,606.2949 & 971.6640 & debt cleared \\
Adverse operating & 1,589.1399 & 1,606.2949 & 791.8125 & debt cleared \\
Thin depth & --- & --- & 1,680.9524 & shortfall 299.2882 \\
Zero liquidity & --- & --- & 2,380.9524 & shortfall 502.6907 \\
\bottomrule
\end{tabular}
\end{center}
The table is generated from the same deterministic registry distributed in CSV and JSON\@. The values verify the controller implementation; because the inputs are author-specified, they are not empirical evidence about venue behavior.

\subsection{Exact leverage tiers}

A finite entry ask ladder is combined with the adverse operating exit envelope. For leverage tiers
\begin{equation}
\mathcal L_{\mathrm{policy}}=\{1,1.25,1.5,1.75,2,2.25,2.5,3\},
\end{equation}
the verifier computes $q_0(L)$, residual entry cash, debt-service upper bound, robust full-sale proceeds, buffer, and $\Psi(L)$. The largest non-negative tier is $L_{\max}^{\mathrm{exact}}$. The linear recovery-ratio formula is calculated alongside it and labeled as a benchmark only.

\subsection{Aggregate liquidity test}

Three positions share one finite bid book. The verifier calculates:
\begin{enumerate}[label=(\roman*)]
  \item the sum of three independently re-anchored liquidation values;
  \item the single aggregate proceeds curve for the combined quantity;
  \item $\Gamma_i$ for every position;
  \item the frozen lowest-$\Gamma_i$ priority order; and
  \item whether every loan clears under the shared-book path without depth reuse.
\end{enumerate}
The independent sum is strictly larger than aggregate proceeds, demonstrating the double-counting error in \Cref{prop:double-counting}. A separate insufficient-capital fixture verifies the pro-rata reserve rule in \eqref{eq:pro-rata-reserve}.

\subsection{Finality checks}

For every successful operating scenario, the verifier evaluates
\begin{equation}
\pi_h\in\left\{0,\frac12,1\right\},
\qquad
\Delta_f\in\{0,1,7,30,365\}\text{ days},
\end{equation}
and several redemption latencies. Loan-channel credit loss remains zero after confirmed debt extinction; user claim value changes with $\pi_h$.

\subsection{Machine-check summary}

The release report lists each invariant individually and fails the build if any required check is false. Negative controls are required to remain negative: thin depth, zero liquidity, premature close, signer loss, collateral escape, and persistent settlement failure must not be mislabeled as successful robust deleveraging.

\subsection{Interpretation}

Deterministic verification proves that the code implements the stated algebra on the registered fixtures. It does not prove:
\begin{itemize}
  \item that the operating set contains the true venue distribution;
  \item that market makers will leave quotes in place during stress;
  \item that settlement will meet the chosen horizon;
  \item that the signer or lien is enforceable; or
  \item that aggregate OI limits are commercially attractive.
\end{itemize}
Those are empirical, contractual, and implementation questions. In particular, calibration and out-of-sample validation of $\Uop$ and $\Ustress$ are reserved for a separate empirical study rather than inferred from the deterministic fixtures.

\section{Reference Architecture and Venue Capability Boundary}
\label{sec:architecture}

The formal mechanism can be implemented without operating a new public CLOB, but not without a credit ledger, execution authority, and settlement-grade reconciliation.

\subsection{Architecture overview}

\begin{figure}[H]
\centering
\resizebox{0.98\textwidth}{!}{%
\begin{tikzpicture}[
node distance=0.75cm and 1.0cm,
every node/.style={font=\small},
box/.style={draw, rounded corners, align=center, minimum height=0.8cm, minimum width=2.8cm},
arrow/.style={-{Latex[length=2mm]}, thick}
]
\node[box] (ui) {Trading terminal\\and portfolio};
\node[box, below=of ui] (api) {Application API\\auth, quotes, positions};
\node[box, below left=of api] (risk) {Risk engine\\robust envelopes\\aggregate capacity};
\node[box, below=of api] (ledger) {Double-entry ledger\\outbox and audit};
\node[box, below right=of api] (exec) {Execution orchestrator\\hard-flat policy};
\node[box, below=of risk] (market) {Market-data engine\\REST/WS/chain resync};
\node[box, below=of ledger] (workers) {Lifecycle, settlement,\\finality and redemption workers};
\node[box, below=of exec] (signer) {Restricted signer\\KMS/MPC or margin vault};
\node[box, below=1.0cm of workers, minimum width=4.6cm] (venue) {Underlying event venue\\order book, vaults, settlement, oracle};

\draw[arrow] (ui) -- (api);
\draw[arrow] (api) -- (risk);
\draw[arrow] (api) -- (ledger);
\draw[arrow] (api) -- (exec);
\draw[arrow] (risk) -- (market);
\draw[arrow] (ledger) -- (workers);
\draw[arrow] (exec) -- (signer);
\draw[arrow] (market) -- (venue);
\draw[arrow] (workers) -- (venue);
\draw[arrow] (signer) -- (venue);
\end{tikzpicture}%
}
\caption{Reference architecture. The venue provides spot execution and final settlement; \Axient supplies credit, robust execution control, aggregate risk limits, accounting, and state transitions.}
\label{fig:architecture}
\end{figure}
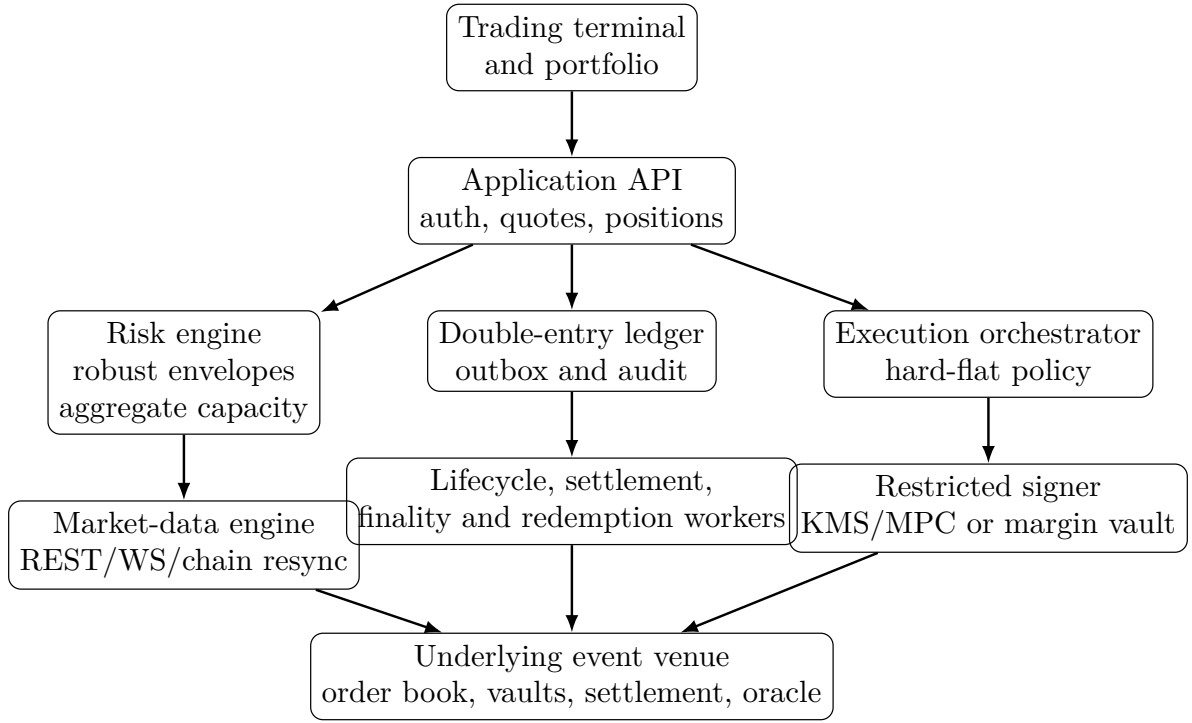

\subsection{Modular-monolith deployment}

The reference implementation should use a modular monolith with separate API and worker processes rather than prematurely distributed microservices. Core modules are:
\begin{itemize}
  \item accounts and controlled subaccounts;
  \item venue adapter and market-data reconstruction;
  \item exact-decimal risk engine;
  \item order and settlement orchestrator;
  \item immutable double-entry ledger;
  \item hard-flat and liquidation workers;
  \item oracle/finality and redemption watcher;
  \item reserve and incident accounting; and
  \item admin and audit console.
\end{itemize}
PostgreSQL is the system of record; a transactional outbox drives asynchronous work. Redis may provide cache and locks but is not authoritative financial state.

\subsection{Illustrative technical stack}

A practical implementation uses TypeScript with strict typing for API and workers, PostgreSQL with fixed-precision numeric columns, Redis for ephemeral coordination, a decimal arithmetic library for off-chain amounts, and a chain client such as viem for EIP-712 and event indexing. Foundry is appropriate for any custom contracts. OpenTelemetry, structured logs, metrics, and Sentry-class error reporting support operations. Every external write is idempotent.

No JavaScript binary floating-point number is used for prices, fees, quantities, or balances. Chain amounts are integers in base units; database amounts are fixed-precision decimals.

\subsection{Venue adapter}

The formal model maps to an interface with methods for market discovery, scheduled close, market status, order-book snapshots and streams, balances, positions, order placement and cancellation, matched fills, settlement confirmation, final payout vector, redemption, and reconciliation. A deterministic simulated adapter implements the same interface for tests.

The adapter exposes the four execution objects separately:
\begin{itemize}
  \item \texttt{quoteBookProceeds};
  \item \texttt{matchedFills};
  \item \texttt{settledTransfers};
  \item \texttt{confirmedRedemption}.
\end{itemize}
No downstream service may infer settlement from a match-only event.

\subsection{Market-data reconstruction}

For a sequence-numbered book stream, the canonical loop is:
\begin{enumerate}[label=\arabic*.]
  \item buffer incoming frames;
  \item fetch a REST snapshot;
  \item apply contiguous deltas by sequence;
  \item mark the book stale on a gap or unexpected silence;
  \item fetch a fresh snapshot; and
  \item block new leveraged quotes until synchronization succeeds.
\end{enumerate}
The risk engine stores the snapshot version and uncertainty-set version used for every quote.

\subsection{Restricted signer}

The signer does not accept arbitrary calldata. It accepts typed intents for approved markets and risk-reducing operations, checks chain, verifying contract, wallet, quantity cap, expiry, nonce, idempotency key, and current risk approval, and writes an audit event before signing.

For a controlled-subaccount implementation, the signer can be KMS or MPC\@. For a trust-minimized version, a venue-recognized margin vault should enforce sale authority, withdrawal lock, debt priority, and post-repayment release on-chain.

\subsection{Current PredictStreet capability mapping}

Table~\ref{tab:predictstreet-capability} is a non-normative reading of public PredictStreet documentation as of July 13, 2026. It is not part of the theorem and should be re-verified before integration.

\begin{table}[H]
\centering
\small
\begin{tabular}{@{}>{\raggedright\arraybackslash}p{0.20\textwidth}>{\raggedright\arraybackslash}p{0.30\textwidth}>{\raggedright\arraybackslash}p{0.40\textwidth}@{}}
\toprule
Capability & Public interface & \Axient implication \\
\midrule
Order-book observation & REST market order book plus live \texttt{token\_book} snapshots and sequence-numbered updates \citep{predictstreet2026marketoverview,predictstreet2026websocket} & Supports $B_u^{\mathrm{book}}$ and resynchronization, but not a guarantee of future settled proceeds. \\
Scheduled close & Market objects expose \texttt{closesAt} and \texttt{kickoff} \citep{predictstreet2026marketoverview} & Supports $T_c^{\mathrm{sched}}$; actual suspension or early close remains distinct. \\
Lifecycle push & Book, settlement, and platform-status channels are live; the condition-lifecycle channel is documented as pending \citep{predictstreet2026websocket} & Poll REST status and reconcile chain events; do not rely on one lifecycle stream. \\
Match versus settlement & Orders can be filled before on-chain settlement; settled trades and failures have separate states \citep{predictstreet2026orderlifecycle} & Use chain-confirmed settlement or accounting-grade trade status for $B^{\mathrm{set}}$. \\
Automatic hard-flat & Every acting wallet must sign EIP-712 writes; multi-wallet partners still need each acting key \citep{predictstreet2026apikeys} & A user-only self-custodial signer cannot be forced to close. A controlled implementation requires MPC-managed subaccounts or venue-supported delegation. \\
Collateral lien & One user vault holds USDC and outcome tokens; emergency withdrawal exists after a timelock \citep{predictstreet2026vaults} & Standard vault does not publicly encode \Axient debt priority or liquidator role. Trust-minimized lending requires a new recognized margin vault or equivalent lien. \\
Dispute and void & Binary markets have proposal/challenge and disputed states; void is a payout-vector outcome \citep{predictstreet2026challenges,predictstreet2026voiddelay} & Compatible with debt-free residual claims; leverage must already be extinguished. \\
Redemption & Gasless redemption progresses asynchronously to confirmed or failure \citep{predictstreet2026redeem} & Model $\tau_r$ separately from $\tau_f$. \\
\bottomrule
\end{tabular}
\caption{Public PredictStreet capability map, accessed July 13, 2026.}
\label{tab:predictstreet-capability}
\end{table}

The table supports a controlled-account implementation but not a claim that the current public vault is a permissionless secured-lending primitive.

\subsection{Smart-contract boundary}

The simulator and partner-integration implementation can remain backend-led because the underlying venue already supplies outcome tokens and settlement. Before meaningful external capital is accepted, useful custom contracts include a restricted capital vault, reserve vault, and fee router with public balances, timelocks, role separation, and allocation caps. These contracts improve transparency but do not create a lien over outcome tokens held in an unrelated venue vault.

A full margin vault is the critical later contract. It must be recognized by the venue's exchange and support:
\begin{enumerate}[label=(\roman*)]
  \item user beneficial ownership;
  \item lender lien while debt is positive;
  \item delegated risk-reducing sale authority;
  \item withdrawal prohibition during debt;
  \item deterministic debt-first proceeds application; and
  \item automatic release after repayment.
\end{enumerate}
Without venue recognition, a custom contract cannot substitute for the venue's signer-to-vault verification or matching operator.

\subsection{Operational modes}

The deployment modes are
\begin{equation}
\state{PAPER}
\rightarrow
\state{SIMULATED}
\rightarrow
\state{PARTNER\_STAGING}
\rightarrow
\state{LIVE\_WHITELIST}
\rightarrow
\state{LIVE\_PUBLIC}.
\end{equation}
Each transition requires explicit gates for reconciliation, signer control, settlement latency, shared-book capacity, reserve, and incident recovery. Public access is not a default consequence of successful simulation.

\section{Security Properties and Trust Model}
\label{sec:security}

The mechanism has mathematical, software, contractual, and venue-level dependencies. Security claims must identify which layer enforces each property.

\subsection{Actors}

Relevant actors are the user, capital provider, \Axient operator, signer administrator, underlying venue operator, market maker, oracle proposer/challenger, chain validators, smart-contract administrator, and reserve custodian. One legal entity may fill several roles, but the risk analysis keeps them separate.

\subsection{Safety properties}

The reference implementation targets:
\begin{enumerate}[label=(S\arabic*)]
  \item no debt in ordinary finality states;
  \item no debt reduction from match-only events;
  \item no duplicate financial posting;
  \item no withdrawal or transfer of pledged assets while debt is positive;
  \item no arbitrary signer operation outside policy;
  \item no position-level reuse of shared-book liquidity;
  \item no final payout accounting from a provisional proposal;
  \item no hardcoded void payout; and
  \item explicit failure-state transition when the robust certificate no longer applies.
\end{enumerate}

\subsection{Liveness properties}

The system should eventually:
\begin{enumerate}[label=(L\arabic*)]
  \item resynchronize a stale book or disable new leverage;
  \item settle, retry, or classify every matched trade;
  \item clear debt, declare shortfall, or invoke a backstop by the hard-flat deadline;
  \item reconcile ledger and venue balances;
  \item ingest a final payout vector; and
  \item confirm redemption or expose a permanent failure state.
\end{enumerate}
Safety takes priority over liveness: failure to obtain a fresh book blocks new leverage rather than using stale marks.

\subsection{Threats and controls}

\begin{longtable}{@{}p{0.25\textwidth}p{0.64\textwidth}@{}}
\toprule
Threat & Principal controls \\
\midrule
\endhead
Book manipulation or quote withdrawal & Robust lower envelope, snapshot versioning, depth caps, aggregate OI limits, maker/RFQ diversification, no midpoint credit. \\
Settlement delay or failure & Match/settlement separation, debt-service upper envelope, finite horizon, pending accounts, retries, failure-state transition. \\
Signer compromise & Typed intents, market and quantity allowlists, KMS/MPC policy, dual approval for configuration, key rotation, immutable audit. \\
Borrower refuses liquidation & Controlled signer or enforceable delegated sale right; otherwise leverage is not offered. \\
Borrower removes collateral & Withdrawal lock/lien or controlled subaccount; continuous balance reconciliation; emergency-path monitoring. \\
Premature market closure & Early buffer, status polling, platform-status feed, partner SLA where available, reserve; no universal guarantee claimed. \\
Shared-book run & Aggregate capacity, deterministic priority, simultaneous-close scenario, market-level circuit breaker. \\
Oracle error or dispute & Debt extinction before finality, payout-vector ingestion, challenge process, no provisional settlement. \\
Redemption failure & Separate $\tau_r$, retry and direct on-chain fallback where available, incident state, no claim of immediate cash. \\
Ledger corruption & Append-only journal, idempotency, database constraints, independent reconciliation, backups and restore tests. \\
Operator insolvency & Segregated assets/reserves, role separation, on-chain transparency where possible; legal segregation remains outside pure code. \\
\bottomrule
\end{longtable}

\subsection{Manipulation and strategic liquidity withdrawal}

The operating uncertainty set should not assume passive independent noise. Market makers may withdraw bids when they anticipate hard-flat flow. The controller therefore includes adversarial depth reduction and simultaneous liquidation in its scenarios. No finite deterministic set can cover every strategic response. Markets that require ordinary use of the reserve to close positions are ineligible for leverage rather than ``insured.''

Leverage also amplifies informed-trading and outcome-manipulation rents documented in the companion research \citep{nechepurenko2026manipulation}. Debt-free finality protects the loan after hard-flat; it does not make the underlying event fair or manipulation-resistant.

\subsection{Trust-minimization ladder}

\begin{center}
\begin{tabular}{@{}>{\raggedright\arraybackslash}p{0.09\textwidth}>{\raggedright\arraybackslash}p{0.27\textwidth}>{\raggedright\arraybackslash}p{0.51\textwidth}@{}}
\toprule
Tier & Architecture & Trust statement \\
\midrule
T0 & Deterministic simulator & No real assets; validates code and algebra only. \\
T1 & Controlled subaccounts, backend ledger & Physically backed but custodial/semi-custodial; users trust operator accounting, signer, and reserve. \\
T2 & On-chain capital/reserve vaults plus controlled venue accounts & Capital and reserve become publicly auditable; outcome-token lien remains operational rather than venue-enforced. \\
T3 & Venue-recognized margin vault & On-chain lien, delegated liquidation, withdrawal lock, and debt-first waterfall materially reduce operator trust. \\
T4 & Multi-venue, independently executable claims & Strongest portability; requires standardized venue and settlement interfaces not assumed here. \\
\bottomrule
\end{tabular}
\end{center}

The base reference implementation is T1. It should be described as wallet-native or on-chain-settled where accurate, not as fully permissionless or trustless.

\subsection{Smart-contract invariants for later phases}

A margin vault should enforce:
\begin{equation}
D>0\Longrightarrow
\begin{cases}
\text{withdrawable pledged quantity}=0,\\
\text{risk-increasing approval}=0,\\
\text{risk-reducing executor enabled},\\
\text{sale cash applied debt-first}.
\end{cases}
\end{equation}
After confirmed repayment,
\begin{equation}
D=0\Longrightarrow\text{residual assets releasable to the user}.
\end{equation}
Formal verification of those contract invariants would complement, not replace, venue and oracle risk analysis.

\subsection{Incident classes}

P0 incidents include loss of signer control, ledger imbalance, unauthorized withdrawal, double settlement, or reserve insolvency. P1 incidents include unexpected close with debt, prolonged settlement pending, aggregate capacity breach, and incorrect payout ingestion. Every incident records affected accounts, external event identifiers, financial exposure, containment, recovery, and postmortem actions.

\section{Limitations}
\label{sec:limitations}

\subsection{No empirical calibration of the operating set}

The lower settled-proceeds envelope, $\Uop$, and the broader $\Ustress$ registry are author-specified in the deterministic verifier. The paper does not estimate how often a real venue leaves the operating set, how settlement latency is distributed, or how market makers respond to anticipated hard-flat flow. Empirical construction, calibration, target-coverage selection, and out-of-sample falsification of these sets are reserved for a separate study; they are the most important empirical inputs for deployment.

\subsection{Robustness depends on set design}

A robust theorem can be technically correct and economically weak if $\Uop$ is narrow, implausible, or selected after observing outcomes. A production policy must pre-register both the operating-set construction and the stress-transform catalogue, evaluate coverage on a disjoint holdout window, and report every breach without retroactively enlarging the set. The present paper supplies the mathematical interface and deterministic test registry, not a calibrated venue model.

\subsection{No optimal execution theorem}

The hard-flat policy certifies debt clearance; it does not minimize expected cost, market impact, user disutility, or information leakage. More sophisticated recourse may preserve additional residual exposure but introduces model and implementation complexity.

\subsection{Venue control remains external}

The mechanism does not control listing, early suspension, order acceptance, matching priority, settlement batching, oracle administration, or redemption. A partner contract can narrow some risks but cannot make an external venue permissionless.

\subsection{Control-and-lien assumptions are substantive}

A controlled subaccount satisfies execution authority operationally but introduces custody and operator risk. A standard user vault without lender lien or delegated liquidation does not satisfy the trust-minimized version of the assumptions. The paper does not solve the legal insolvency treatment of operationally controlled collateral.

\subsection{Scheduled close is not actual close}

A published close time supports planning but does not rule out early suspension or platform freeze. The robust theorem requires actual tradability through the horizon; premature closure remains a failure path unless an enforceable guarantee exists.

\subsection{Aggregate liquidity can be strategic}

The shared-book model prevents mechanical double counting but does not model equilibrium response. Market makers may cancel, front-run, or reprice when aggregate hard-flat flow is predictable. The scenario set can stress this behavior but cannot prove a game-theoretic equilibrium.

\subsection{No cross-margin or portfolio netting}

The base mechanism uses isolated risk buckets. Payoff-algebraic offsets and portfolio margin may improve capital efficiency, but they complicate aggregate execution, priority, and finality and are outside scope.

\subsection{No welfare theorem}

Debt-free finality can protect lender principal while imposing forced-sale cost and reducing user residual exposure. The paper does not prove that the mechanism maximizes trader welfare or social information aggregation.

\subsection{Dispute invariance is role-specific}

The result applies to lender credit-principal exposure through the extinguished loan. An operator can remain exposed through custody, SLA, legal, reputational, oracle-governance, or redemption-guarantee roles.

\subsection{Reserve sufficiency is scenario-conditional}

A reserve bound is only as complete as the scenario set, shared-book model, asset availability, and legal segregation. It is not a universal solvency theorem.

\subsection{Deterministic verification is not production proof}

Passing fixtures demonstrates code/formula consistency. It does not establish API uptime, contract correctness, key security, settlement finality, or commercial liquidity.

\subsection{Manipulation and informed trading remain}

The mechanism does not prevent insiders from trading, actors from influencing event outcomes, or traders from manipulating visible book states. It limits loan exposure after successful hard-flat; it does not make the market epistemically clean.

\subsection{Legal characterization is outside scope}

The instrument may be characterized as margin lending over spot outcome tokens, a derivative, gaming, or a combination depending on jurisdiction and implementation. This paper is not legal advice.

\subsection{Authorial conflict of interest}

The author is developing \Axient and benefits from interest in the mechanism. The paper mitigates but cannot eliminate this conflict through explicit assumptions, negative controls, source release, and a separation between theorem, verifier, and deployment claim.

\subsection{What would falsify the practical thesis}

The practical thesis would be weakened or falsified if controlled evidence shows that, at commercially meaningful size:
\begin{enumerate}[label=(\roman*)]
  \item the registered lower envelope is breached too often to support positive unit economics;
  \item aggregate bid capacity is routinely below financed debt;
  \item actual close or settlement failure frequently precedes hard-flat;
  \item required signer/lien control is unavailable under acceptable trust; or
  \item fees and reserve capital exceed user willingness to pay for leverage.
\end{enumerate}
The formal debt-free-finality invariant would remain true after confirmed repayment, but the product architecture would not be practically useful.

\section{Conclusion}
\label{sec:conclusion}

\Axient separates a leveraged event position into two maturities. The loan ends when settled sale proceeds extinguish debt; the residual outcome claim ends only when the payout vector is final and redemption confirms. This separation removes terminal payout and dispute duration from the lender's loan channel after confirmed conversion.

The central revision is ex ante. A visible book and a realized settlement curve are not the same object. The planned sale is therefore selected from a lower envelope of future settled proceeds over a registered operating uncertainty set, against an upper envelope of debt. The realized minimum sale remains useful as an audit and residual-exposure measure, but it no longer circularly determines the earlier order.

The same discipline applies at system scale. Multiple positions cannot each reuse the same top-of-book liquidity. Aggregate OI, lowest-$\Gamma_i$ hard-flat priority, and pro-rata principal-shortfall reserve allocation are computed from shared execution paths. The exact book-dependent leverage envelope is primary; scalar recovery-ratio formulas are only benchmarks.

The resulting guarantees are deliberately conditional. Inside the registered operating set, with enforceable execution authority, collateral non-escape, and settlement by the horizon, robust auto-deleverage clears debt. Outside that set, no backend-only mechanism with leverage above one can promise zero shortfall without external collateral, liquidity guarantees, or reserves. This boundary is a feature of the specification, not a defect to be hidden.

The next research step is empirical rather than algebraic: estimate settled-proceeds envelopes, actual-close risk, shared-book capacity, and strategic liquidity response on a partner venue, pre-register coverage thresholds, and test whether commercially useful leverage remains after all constraints. Until then, the paper supports a proof of mechanism and intent, not a claim of production safety.

\appendix
\section{Notation, Timelines, and State Glossary}
\label{app:notation}

This appendix consolidates the notation used in the main text. Quantities called ``quoted,'' ``matched,'' ``settled,'' and ``redeemed'' are intentionally distinct.

\subsection{Core symbols}

\begin{longtable}{@{}p{0.20\textwidth}p{0.70\textwidth}@{}}
\toprule
Symbol & Definition \\
\midrule
\endfirsthead
\toprule
Symbol & Definition \\
\midrule
\endhead
$(\Omega,\F,(\F_t),\Prob)$ & Filtered probability space carrying venue, execution, settlement, oracle, and account observations. \\
$E$ & Binary real-world event. \\
$h\in\{\yes,\no\}$ & Outcome token held by the position. Long-NO is the physically backed implementation of short-YES. \\
$\boldsymbol\pi=(\pi_{\yes},\pi_{\no})$ & Final payout vector, with entries in $[0,1]$ and $\pi_{\yes}+\pi_{\no}=1$. \\
$t_0$ & Confirmed position-open time. \\
$\sigma_R$ & Reduce-only transition time. \\
$u$ & Hard-flat decision and execution-start time. \\
$\nu$ & Time of the last matcher fill used by the hard-flat execution sequence. \\
$\sigma$ & First confirmed time at which settled proceeds have been applied and debt is zero; $\sigma=\infty$ on failure to extinguish debt. \\
$T_c^{\mathrm{sched}}$ & Venue-published scheduled close time. \\
$T_c^{\mathrm{act}}$ & Actual time after which risk-reducing orders are no longer accepted. \\
$\tau_p$ & First provisional oracle proposal time. \\
$\tau_f$ & Time at which the final payout vector is authoritative for redemption. \\
$\tau_r$ & Time at which residual outcome-token redemption is confirmed. \\
$C$ & User collateral contributed at entry. \\
$L$ & Gross leverage, equal to acquisition budget divided by user collateral. \\
$N=LC$ & Gross acquisition budget. \\
$D_0=(L-1)C$ & Initial borrowed principal. \\
$D_t$ & Outstanding debt before extinction. \\
$r_t$ & Borrowing-rate process. \\
$C_t^{\mathrm{debt}}$ & Charges contractually added to the receivable. \\
$K_t$ & Dedicated confirmed cash available to the position and loan waterfall. \\
$q_0(L)$ & Outcome-token quantity acquired and settled at entry for leverage tier $L$. \\
$q$ & Settled outcome-token quantity at a hard-flat decision time. \\
$q_{\mathrm{spot}}$ & Residual fully funded outcome-token quantity after debt extinction. \\
$A_{t_0}^{\mathrm{set}}(q)$ & Confirmed cumulative acquisition cost of $q$ tokens, including allocated entry costs. \\
$B_u^{\mathrm{book}}(x)$ & Net proceeds implied by the observed bid book at decision time $u$; a quote, not settled cash. \\
$B_{u,\Delta}^{\mathrm{match},\Pi}(x,\omega)$ & Net proceeds recorded by matched fills under execution policy $\Pi$ by horizon $u+\Delta$. \\
$B_{u,\Delta}^{\mathrm{set},\Pi}(x,\omega)$ & Net proceeds confirmed and available for debt application by $u+\Delta$. \\
$B_{\sigma}^{\mathrm{audit}}(x,\omega)$ & Cumulative confirmed proceeds of the first $x$ units in the realized fill sequence; computed after settlement. \\
$\Uop_u(\Delta)$ & Registered operating uncertainty set over the hard-flat horizon; the robust debt-clearing certificate applies only on this set. \\
$\Ustress_u(\Delta)$ & Broader registered stress set containing $\Uop_u(\Delta)$ plus named adversarial and operational scenarios used for validation and reserve design; stress inclusion is not a guarantee. \\
$\Ufail_u(\Delta)$ & Full failure set, defined as $\Omega\setminus\Uop_u(\Delta)$; includes premature close, zero liquidity, lost authority, or persistent settlement failure. \\
$T_{k,h,s}(\mathcal L_u)$ & Adversarial bid-book transform: remove the best $k$ levels, remove depth fraction $h$, and expand bid distance from the reference mid by factor $s$. \\
$\underline B_{u,\Delta}^{\Pi}(x)$ & Lower settled-proceeds envelope over $\Uop_u(\Delta)$. \\
$\overline H_{u,\Delta}$ & Upper bound on debt by the settlement horizon over $\Uop_u(\Delta)$. \\
$\Qset_u(q)$ & Finite admissible cumulative sale grid generated by balances, token precision, lot rules, and policy. \\
$m_u$ & Explicit hard-flat buffer in addition to the debt-service upper bound. \\
$\widehat{\mathcal F}_{u,\Delta}(q)$ & Ex-ante robust debt-clearing feasible set. \\
$\widehat x_u$ & Minimum quantity certified ex ante to clear debt and buffer over the operating set. \\
$\mathcal F_{\sigma}^{\mathrm{audit}}$ & Realized post-settlement debt-clearing feasible set. \\
$x_{\sigma}^{\ast}$ & Minimum realized quantity whose confirmed proceeds suffice to clear debt. \\
$Z_{\sigma}$ & Confirmed cash remaining after debt application. \\
$\underline V_t(q)$ & Conservative executable liquidation value used for monitoring. \\
$E_t(q)$ & Executable equity $K_t+\underline V_t(q)-D_t$. \\
$M_t(q)$ & Maintenance and forward-looking execution buffer. \\
$\hf_t(q)$ & Health factor $(K_t+\underline V_t(q))/(D_t+M_t(q))$. \\
$\Psi(L)$ & Exact finite-book feasibility margin for candidate leverage tier $L$. \\
$\mathcal L_{\mathrm{policy}}$ & Finite set of product-permitted leverage tiers. \\
$L_{\max}^{\mathrm{exact}}$ & Largest tier with non-negative exact robust feasibility margin. \\
$a,b,g=b/a$ & Effective entry unit cost, stressed exit unit recovery, and scalar recovery ratio in the linear benchmark only. \\
$\mathcal U_{u,\mathrm{agg}}^{\mathrm{op}}(\Delta)$ & Registered aggregate operating set for positions sharing one execution book or risk bucket. \\
$\mathcal U_{u,\mathrm{agg}}^{\mathrm{stress}}(\Delta)$ & Broader aggregate stress set used for correlated hard-flat, common-delay, and shared-liquidity tests. \\
$\underline{\mathcal B}_{u,\Delta}(X)$ & Shared lower settled-proceeds envelope for aggregate quantity $X$ through one common book. \\
$\Gamma_i(u)$ & Standalone robust coverage ratio used only for deterministic priority ordering; aggregate solvency still uses shared-book increments. \\
$\Psi_{\mathrm{agg}}(\mathbf q)$ & Aggregate risk-bucket coverage margin after shared-book execution and buffers. \\
$\ell_k(s)$ & Aggregate credit shortfall of risk cluster $k$ in scenario $s$. \\
$r_i(R,\boldsymbol\ell)$ & Reserve allocation to position $i$ under the pro-rata principal-shortfall rule. \\
$\mathcal R$ & Segregated reserve legally and operationally available to absorb modeled shortfalls. \\
$\Sset$ & Registered finite reserve stress-scenario set. \\
$V_{\tau_f}^{\mathrm{claim}}$ & User's final claim value before redemption. \\
$W_{\tau_r}^{\mathrm{cash}}$ & Confirmed user cash after redemption. \\
\bottomrule
\end{longtable}

\subsection{Economic clocks}

The successful path has the ordering
\begin{equation}
 t_0<\sigma_R\le u\le\nu\le\sigma<T_c^{\mathrm{act}}\le\tau_p\le\tau_f\le\tau_r.
\end{equation}
The mechanism does not assume that $T_c^{\mathrm{act}}$, $\nu$, $\sigma$, $\tau_f$, or $\tau_r$ are deterministic. A scheduled timestamp is a planning input; a settlement or finality time is a stopping time defined by observed state.

\subsection{Canonical account states}

\begin{longtable}{@{}p{0.31\textwidth}p{0.59\textwidth}@{}}
\toprule
State & Meaning and debt rule \\
\midrule
\endfirsthead
\toprule
State & Meaning and debt rule \\
\midrule
\endhead
\state{PENDING\_OPEN} & Entry has been submitted, matched, or is awaiting settlement. A final financed position does not yet exist. \\
\state{OPEN} & Settled leveraged position. Positive debt is permitted subject to risk and aggregate-capacity limits. \\
\state{REDUCE\_ONLY} & Risk may be reduced or collateral added; no risk-increasing action is permitted. Positive debt may remain. \\
\state{DELEVERAGING} & Hard-flat orders and settlement are in progress. Positive debt remains until confirmed application. \\
\state{EVENT\_CLOSE\_EXCEPTION} & Actual venue closure occurred before ordinary debt extinction. Positive debt may remain and the position is outside the robust certificate. \\
\state{CREDIT\_SHORTFALL} & Settled proceeds and dedicated cash are insufficient. The uncovered amount is explicit and enters recovery or reserve accounting. \\
\state{DEBT\_FREE\_SPOT} & Debt is zero; residual outcome tokens are fully funded spot claims. \\
\state{PENDING\_FINALITY} & Trading is closed and the payout vector is not final. Ordinary positions must have zero debt. \\
\state{DELAYED} & Resolution is delayed because the source or process cannot safely finalize. Ordinary positions must have zero debt. \\
\state{DISPUTED} & A provisional result is challenged. Ordinary positions must have zero debt. \\
\state{FINAL} & Payout vector is final. Redemption may still be pending. \\
\state{REDEMPTION\_PENDING} & A redemption request has not yet reached final confirmation. \\
\state{REDEEMED} & Outcome token has been converted into confirmed collateral and user accounting is complete. \\
\bottomrule
\end{longtable}

\subsection{Interpretation convention}

A symbol indexed by ``book'' is an observable quote; ``match'' is a matcher record; ``set'' is confirmed collateral available to the debt waterfall; ``audit'' is a post-settlement reconstruction; and ``redeem'' is final conversion of an outcome claim into cash. No theorem permits one category to be substituted silently for another.

\section{Supplementary Proofs and Extensions}
\label{app:proofs}

\subsection{Monotonicity of lower settled-proceeds envelopes}

\begin{lemma}[Envelope monotonicity on the admissible grid]
\label{lem:envelope-monotone}
Suppose that for every $\omega\in\Uop_u(\Delta)$, the settled-proceeds function $B_{u,\Delta}^{\mathrm{set},\Pi}(\cdot,\omega)$ is non-decreasing on the common admissible grid $\Qset_u(q)$. Then the lower envelope $\underline B_{u,\Delta}^{\Pi}$ in \eqref{eq:lower-proceeds-envelope} is non-decreasing on that grid.
\end{lemma}

\begin{proof}
For admissible $x\le y$, every path satisfies
$B_{u,\Delta}^{\mathrm{set},\Pi}(x,\omega)\le B_{u,\Delta}^{\mathrm{set},\Pi}(y,\omega)$. Taking the infimum over the same path set preserves the inequality.
\end{proof}

This lemma is the only monotonicity needed for the finite-grid existence result in \Cref{prop:planned-existence}. Concavity, continuity, differentiability, and a parametric distribution are unnecessary.

\subsection{Stopping-time interpretation of debt extinction}

Define
\begin{equation}
\sigma
=
\inf\left\{t\ge u:
\text{confirmed dedicated cash has been applied and }D_t=0
\right\}.
\label{eq:stopping-time-sigma}
\end{equation}
If settlement confirmations, journal postings, and debt balances are $\F_t$-observable, the event $\{\sigma\le t\}$ is $\F_t$-measurable. Thus $\sigma$ is a stopping time. On a path where debt is never extinguished, the convention $\sigma=\infty$ keeps the state explicit rather than assigning a fictitious leverage maturity.

The definition also clarifies why match time $\nu$ is not leverage maturity. A match may remain pending, fail, or settle for a different net amount after charges. The account reaches the debt-free state only at \eqref{eq:stopping-time-sigma}.

\subsection{Discrete-lot certification overshoot}

Let the admissible robust sale grid be $0=y_0<\cdots<y_m$ and let $\widehat x_u=y_j$ be the minimum certified quantity. Define the marginal lower-envelope proceeds
\begin{equation}
\Delta\underline B_j
=
\underline B_{u,\Delta}^{\Pi}(y_j)
-
\underline B_{u,\Delta}^{\Pi}(y_{j-1}).
\end{equation}

\begin{lemma}[Robust certificate overshoot]
\label{lem:robust-overshoot}
If $j\ge1$, then the certified cash surplus above the debt bound and buffer satisfies
\begin{equation}
0\le
K_u+\underline B_{u,\Delta}^{\Pi}(y_j)
-\overline H_{u,\Delta}-m_u
<\Delta\underline B_j.
\end{equation}
\end{lemma}

\begin{proof}
Feasibility of $y_j$ gives the weak lower bound. Minimality implies
$K_u+\underline B_{u,\Delta}^{\Pi}(y_{j-1})<\overline H_{u,\Delta}+m_u$. Add $\Delta\underline B_j$ to both sides and rearrange.
\end{proof}

An analogous pathwise statement holds for the realized audit minimum $x_\sigma^\ast$ on the cumulative settled-fill grid. Pending matched batches require an additional overshoot bound because cancellation may not prevent their later settlement.

\subsection{Detailed robust-clearing argument}

The proof of \Cref{thm:robust-clearing} can be written as a chain of inequalities for the full certified cap. For the realized operating path $\omega$,
\begin{align}
K_u+B_{u,\Delta}^{\mathrm{set},\Pi}(\widehat x_u,\omega)
&\ge
K_u+\underline B_{u,\Delta}^{\Pi}(\widehat x_u) \\
&\ge
\overline H_{u,\Delta}+m_u \\
&\ge
D_{u+\Delta}(\omega)+m_u.
\label{eq:robust-proof-chain}
\end{align}
Thus the full cap is sufficient to cover realized debt and the registered buffer. The live controller need not always settle the full cap. If it stops at the first confirmed prefix that drives debt to zero, that prefix guarantees only non-negative residual cash; it need not preserve $m_u$. The buffer guarantee requires either an explicit debt-plus-buffer stopping rule or settlement of the complete certified cap. \Cref{ass:authority,ass:lien} make the risk-reducing action and debt-first application enforceable; \Cref{ass:settlement} identifies the value eligible for the left-hand side. The first time the journaled cash application drives debt to zero is the stopping time $\sigma\le u+\Delta$.

The theorem is therefore an implication of three independently auditable claims: a conservative proceeds certificate, a conservative debt certificate, and enforceable control over the financed asset. Failure of any one moves the position outside the theorem's scope.

\subsection{Realized residual-exposure maximality with cash}

Suppose a realized debt-clearing sale $x$ leaves surplus cash
\begin{equation}
Z(x)=K_u+B_\sigma^{\mathrm{audit}}(x)-D_{\sigma^-}\ge0.
\end{equation}
The objective ``preserve the largest event exposure'' orders solutions by residual quantity $q-x$, not by surplus cash. Under that objective, \Cref{prop:residual-maximality} follows immediately from minimum quantity. If instead the user values cash and event exposure through a utility function $U(q-x,Z(x))$, the minimum-quantity rule need not be utility-optimal. The paper makes no welfare-optimality claim.

\subsection{Exact leverage on a finite policy set}

\begin{proposition}[Existence of the exact tier cap]
\label{prop:finite-tier-cap}
If $\mathcal L_{\mathrm{policy}}$ is finite and contains at least one $L$ for which $\Psi(L)\ge0$, then $L_{\max}^{\mathrm{exact}}$ in \eqref{eq:exact-leverage-envelope} exists and is unique as the largest numerical tier. If no tier is feasible, no new leveraged position is admitted.
\end{proposition}

\begin{proof}
The feasible subset of a finite totally ordered set is finite. Every non-empty finite totally ordered set has a unique maximum.
\end{proof}

The result does not require $\Psi(L)$ to be globally monotone. Entry ladders, fee tiers, discrete quantities, and depth cliffs can create local non-monotonicity. A reference implementation should evaluate every approved tier or establish monotonicity before using bisection.

\subsection{Linear-execution benchmark}

Under proportional entry cost $A(q)=aq$, proportional robust exit proceeds $\underline B(q)=bq$, no residual entry cash, debt growth factor $\rho_D\ge1$, and zero explicit buffer, candidate leverage $L$ is feasible exactly when
\begin{equation}
gL\ge\rho_D(L-1),
\qquad g=\frac{b}{a}.
\end{equation}
When $0\le g<\rho_D$, equality gives
\begin{equation}
L_{\max}^{\mathrm{lin}}=\frac{\rho_D}{\rho_D-g}.
\end{equation}
With a buffer of $mC$, the corresponding bound is
\begin{equation}
L\le\frac{\rho_D-m}{\rho_D-g}
\end{equation}
when the numerator and denominator are positive. These formulas are algebraic benchmarks only. The exact controller uses \eqref{eq:exact-feasibility-margin} and shared-book constraints.

\subsection{Top-of-book double counting from marginal proceeds}

Let $p(z)$ be a non-increasing non-negative marginal net proceeds curve and
\begin{equation}
B(x)=\int_0^x p(z)\,\dd z.
\end{equation}
For $x,y\ge0$,
\begin{align}
B(x+y)-B(x)
&=\int_x^{x+y}p(z)\,\dd z \\
&\le\int_0^y p(z)\,\dd z=B(y).
\end{align}
Hence $B(x+y)\le B(x)+B(y)$. Iteration gives \eqref{eq:double-counting}. On a finite step book, the same proof is a sum over marginal lots.

The inequality is strict when some of the aggregate marginal units trade at a worse net price than the separately re-anchored units. This is the typical case when the combined sale walks beyond the best levels.

\subsection{Priority-allocation telescoping}

For priority order $\pi$, let $X_j=\sum_{k\le j}x_{\pi_k}$ and define the pathwise increment
\begin{equation}
g_j(\omega)
=
\mathcal B_{u,\Delta}^{\mathrm{set},\Pi}(X_j,\omega)
-
\mathcal B_{u,\Delta}^{\mathrm{set},\Pi}(X_{j-1},\omega).
\end{equation}
For every path,
\begin{equation}
\sum_{j=1}^{n}g_j(\omega)
=
\mathcal B_{u,\Delta}^{\mathrm{set},\Pi}(X_n,\omega),
\end{equation}
with $X_0=0$. The certified incremental amount is $\underline g_j=\inf_{\omega}g_j(\omega)$ as in \eqref{eq:incremental-proceeds}. Therefore
\begin{equation}
\sum_{j=1}^{n}\underline g_j
\leq
\inf_{\omega}
\mathcal B_{u,\Delta}^{\mathrm{set},\Pi}(X_n,\omega)
=
\underline{\mathcal B}_{u,\Delta}(X_n).
\end{equation}
Thus the lower-bound allocation is conservative and assigns every unit of pathwise aggregate proceeds at most once. The priority-feasibility inequalities in \eqref{eq:priority-feasibility} determine whether each segregated loan clears under that shared execution family.

\subsection{Adversarial construction for the impossibility boundary}

Fix any backend algorithm that permits $L>1$ and therefore creates $D_0>0$. Consider two paths that are identical up to the instant before the first debt-clearing order can become irrevocably effective:
\begin{enumerate}[label=(\roman*)]
  \item on the operating path, the venue remains open, the signer remains available, and sufficient bids settle;
  \item on the failure path, the venue closes, all bids disappear, the signer authority is revoked, or the pledged asset is removed.
\end{enumerate}
The algorithm cannot distinguish the paths before the branching event. On the failure path no dedicated sale cash is produced. If the held token later pays zero and no independent collateral or guarantee covers the balance, the lender shortfall is positive by \Cref{prop:positive-debt-impossibility}. This proves \Cref{thm:no-universal} without assuming any particular quote process.

\subsection{State invariant with exception states}

Let
\begin{equation}
\mathcal B=
\{\state{OPEN},\state{REDUCE\_ONLY},\state{DELEVERAGING},
\state{EVENT\_CLOSE\_EXCEPTION},\state{CREDIT\_SHORTFALL}\}
\end{equation}
be the states in which positive debt may exist. A failed hard-flat transition is
\begin{equation}
\state{DELEVERAGING}\longrightarrow\state{CREDIT\_SHORTFALL}\in\mathcal B,
\end{equation}
not a transition into an ordinary finality state. The underlying market may already be pending resolution, but the account state must preserve the unresolved credit exception. This prevents a lifecycle label from masking a positive receivable.

\subsection{Reserve sufficiency under several collateral assets}

If cluster shortfalls are denominated in several assets, let $p_j(s)$ be the conservative conversion value of asset $j$ into the reserve numeraire in scenario $s$. Define
\begin{equation}
\ell_k^{\mathrm{num}}(s)=\sum_j p_j(s)\ell_{kj}(s).
\end{equation}
The reserve condition becomes
\begin{equation}
\mathcal R
\ge
\sup_{s\in\Sset}
\sum_k\sum_j p_j(s)\ell_{kj}(s).
\end{equation}
The conversion values must themselves be stressed. Treating depegging collateral at par can make a formal reserve inequality economically meaningless.

\section{Reference Algorithms}
\label{app:algorithms}

The pseudocode is normative at the invariant level and illustrative at the API level. An implementation may use different queues, languages, or venue calls, but it should preserve the separation among quote, match, settlement, debt application, finality, and redemption.

\subsection{Entry admission and exact leverage tiers}

\begin{lstlisting}[language={},caption={Exact finite-book entry admission}]
function admitEntry(request, account, market):
    assert market.state == OPEN
    assert marketData.isSynchronized(market)
    assert liquidationAuthority.isEnforceable(account)
    assert collateralPolicy.preventsEscape(account)

    candidates = sortAscending(product.allowedLeverageTiers)
    feasible = []

    for L in candidates:
        budget = request.collateral * L
        entry = venue.settledBuyQuote(request.heldToken, budget)
        q0 = entry.quantity
        K0 = budget - entry.confirmedCost

        u = lifecycle.hardFlatDecisionTime(market, L)
        Delta = risk.registeredSettlementHorizon(market, L)
        Uop = risk.operatingSet(market, request.heldToken, u, Delta)
        Ustress = risk.registeredStressSet(market, request.heldToken, u, Delta)
        assert Uop is subset of Ustress
        lowerExit = risk.lowerSettledProceedsEnvelope(
            policy = hardFlatPolicy,
            scenarios = Uop,
            quantity = q0
        )
        debtUpper = credit.upperDebtEnvelope(
            principal = (L - 1) * request.collateral,
            horizon = u + Delta
        )
        buffer = risk.hardFlatBuffer(account, market, q0, L)

        psi = K0 + lowerExit - debtUpper - buffer
        aggregatePsi = risk.aggregateCoverageAfterHypotheticalEntry(
            account, market, q0, debtUpper, buffer
        )

        if psi >= 0 and aggregatePsi >= 0:
            feasible.append(L)

    assert request.leverage in feasible
    return buildTransparentEntryPreview(request, feasible)

# The maximum admissible leverage is max(feasible), not a midpoint ratio.
# The current book is a quote input, not a settlement guarantee.
# The loan activates only after the entry acquisition settles.
\end{lstlisting}

\subsection{Robust hard-flat planner}

\begin{lstlisting}[language={},caption={Ex-ante robust sale planner}]
function planHardFlat(position, decisionTime):
    assert position.state in {REDUCE_ONLY, DELEVERAGING}
    assert position.debt > 0
    assert venue.acceptsRiskReducingOrders(position.market)
    assert signer.authorizedForRiskReduction(position.account)
    assert collateral.isLockedToDebtWaterfall(position)

    Delta = risk.registeredSettlementHorizon(position.market)
    Uop = risk.freezeOperatingSet(
        position, decisionTime, Delta, currentBookVersion
    )
    Ustress = risk.freezeStressSet(
        position, decisionTime, Delta, currentBookVersion
    )
    assert Uop is subset of Ustress
    lowerCurve = risk.lowerSettledProceedsCurve(
        executionPolicy, Uop, position.quantity
    )
    debtUpper = credit.upperDebtEnvelope(
        position.debt, decisionTime, Delta, Uop
    )
    target = debtUpper + risk.hardFlatBuffer(position, Uop)

    xHat = first admissible x such that
        position.dedicatedCash + lowerCurve(x) >= target

    if no xHat exists:
        transition(position, CREDIT_SHORTFALL)
        openIncident(position, ROBUST_SET_INFEASIBLE)
        return INFEASIBLE

    persistImmutableCertificate(
        positionId = position.id,
        decisionTime,
        horizon = Delta,
        operatingSetVersion = Uop.version,
        stressSetVersion = Ustress.version,
        bookVersion = currentBookVersion,
        debtUpper,
        target,
        plannedSaleCap = xHat
    )
    return xHat
\end{lstlisting}

The certificate is frozen before execution. Refreshing the book may produce a new certificate, but the implementation must not overwrite the old one or claim retrospectively that an adverse path was inside the earlier operating set.

\subsection{Recourse execution and settlement loop}

\begin{lstlisting}[language={},caption={Hard-flat execution with settlement-confirmed stopping}]
function executeHardFlat(positionId, certificateId):
    lock positionId
    position = loadForUpdate(positionId)
    certificate = loadCertificate(certificateId)
    assert certificate.positionId == positionId
    transition(position, DELEVERAGING)
    cancelRiskIncreasingOrders(position)
    unlock positionId

    remainingCap = certificate.plannedSaleCap

    while position.debt > 0 and remainingCap > 0:
        if venue.actualState(position.market) != OPEN:
            transition(position, EVENT_CLOSE_EXCEPTION)
            openIncident(position, PREMATURE_CLOSE)
            return FAILURE_SET

        intent = executionPolicy.nextRiskReducingIntent(
            position,
            remainingCap,
            freshSynchronizedBook(),
            certificate
        )
        signedOrder = signer.signTypedIntent(intent)
        order = venue.submit(signedOrder, deterministicId(intent))

        wait until order has a terminal settlement classification or retry deadline

        for each newly confirmed settlement event:
            applySettlementAtomically(position, event)
            remainingCap -= event.settledQuantity

        for each settlement failure:
            postCompensatingJournal(event)
            markFailureForReconciliation(event)

        if position.debt == 0:
            cancelUnmatchedRemainder()
            transition(position, DEBT_FREE_SPOT)
            return SUCCESS

        if risk.realizedPathNoLongerSatisfiesCertificate(certificate):
            openIncident(position, OPERATING_SET_BREACH)
            executeFailurePolicy(position)
            return FAILURE_SET

    if position.debt > 0:
        transition(position, CREDIT_SHORTFALL)
        executeReserveAndRecoveryWaterfall(position)
        return SHORTFALL
\end{lstlisting}

A match does not decrement debt. Only \texttt{applySettlementAtomically} can do so.

\subsection{Atomic settlement and debt application}

\begin{lstlisting}[language={},caption={Idempotent settlement handler}]
function applySettlementAtomically(position, settlementEvent):
    if journal.contains(settlementEvent.externalId):
        return ALREADY_APPLIED

    assert settlementEvent.chain == configuredChain
    assert settlementEvent.order belongs to position
    assert settlementEvent.assetTransfers match signed intent

    atomic database transaction:
        move pendingOutcomeQuantity to confirmedSoldQuantity
        credit confirmedNetCash by settlementEvent.netCash
        post venue and protocol fees separately

        available = position.dedicatedCash + settlementEvent.netCash
        payment = min(position.debt, available)
        debit confirmedNetCash / dedicatedCash by payment
        credit loanReceivable by payment
        position.debt -= payment
        position.dedicatedCash = available - payment

        insert unique journal key settlementEvent.externalId
        persist confirmed chain / venue reference

        if position.debt == 0:
            record debtExtinguishmentTime = settlementEvent.confirmedAt

    enqueue independentBalanceReconciliation(position.account)
\end{lstlisting}

No negative debt is manufactured. Cash received after extinction is surplus user cash or reserve cash according to the pre-registered waterfall.

\subsection{Aggregate shared-book planner}

\begin{lstlisting}[language={},caption={Priority-feasible aggregate hard-flat planning}]
function planAggregateHardFlat(riskBucket):
    positions = loadDebtPositions(riskBucket)
    for position in positions:
        position.robustCoverage = (
            position.dedicatedCash
            + risk.standaloneLowerFullSaleValue(position)
        ) / (position.debtUpper + position.buffer)

    priority = sortBy(
        robustCoverage ascending,
        hardFlatDeadline ascending,
        positionId ascending
    )

    sharedScenarioFamily = risk.sharedPathwiseProceedsFamily(
        riskBucket, sum(position.quantity), registeredAggregateScenarios
    )

    consumed = 0
    allocation = []

    for position in priority:
        target = position.debtUpper + position.buffer - position.dedicatedCash
        incrementalLowerCurve = risk.incrementalLowerProceedsCurve(
            sharedScenarioFamily,
            consumedBefore = consumed,
            maxPositionQuantity = position.quantity
        )
        x = first admissible quantity such that
            incrementalLowerCurve(x) >= target

        if no x exists within position.quantity:
            return AGGREGATE_INFEASIBLE

        allocation.append(position.id, x, consumed, incrementalLowerCurve.hash)
        consumed += x

    persistAggregateCertificate(
        bucketVersion,
        scenarioVersion,
        priority,
        allocation,
        sharedScenarioFamily.hash
    )
    return allocation
\end{lstlisting}

This algorithm prevents each position from reusing the top of the same book. The standalone coverage ratio determines ordering only; it is never summed as a solvency measure. If the legal structure permits pooled repayment, the planner may solve one aggregate inequality instead, but the pooling rule must be explicit.

\subsection{Pro-rata principal-shortfall reserve allocation}

\begin{lstlisting}[language={},caption={Deterministic reserve allocation after shared-book execution}]
function allocateReserve(shortfalls, reserveAmount):
    assert every shortfall >= 0
    total = sum(shortfalls)

    if total == 0:
        return zero allocation

    if total <= reserveAmount:
        return allocation[i] = shortfalls[i] for every i

    for position i:
        allocation[i] = reserveAmount * shortfalls[i] / total
        assert 0 <= allocation[i] <= shortfalls[i]

    assert sum(allocation) == reserveAmount
    persistImmutableReserveAllocation(shortfalls, reserveAmount, allocation)
    return allocation
\end{lstlisting}

The rule is applied independently at each waterfall stage. Rounding residuals caused by token or collateral precision are assigned by deterministic position identifier and may not increase total allocation beyond the available reserve.

\subsection{Market lifecycle and finality worker}

\begin{lstlisting}[language={},caption={Lifecycle, dispute, and redemption worker}]
function processMarket(marketId):
    restState = venue.getMarketState(marketId)
    chainState = oracleIndexer.getState(marketId)
    local = loadMarket(marketId)

    if restState == OPEN and now >= local.reduceOnlyAt:
        transitionMarket(local, REDUCE_ONLY)
        rejectRiskIncreasingOrders(local)

    if restState == OPEN and now >= local.hardFlatDecisionAt:
        enqueueAggregateHardFlat(local.riskBucket)

    if restState in {CLOSED, PAUSED, SUSPENDED}:
        rejectAllNewOrders(local)
        if any ordinary position has debt > 0:
            transitionAffectedPositions(EVENT_CLOSE_EXCEPTION)
            openIncident(local, CLOSED_WITH_DEBT)
        else:
            transitionMarket(local, PENDING_FINALITY)

    if chainState == PROPOSED:
        storeProvisionalVector(chainState)
        doNotPostFinalPnL()

    if chainState in {DELAYED, DISPUTED}:
        assert all ordinary positions have debt == 0
        transitionMarket(local, chainState)

    if chainState == FINAL:
        storeFinalPayoutVector(chainState)
        transitionMarket(local, FINAL)
        enqueueRedemptions(local)

    for each redemption job:
        if chain confirmation observed:
            postRedemptionCashIdempotently(job)
            transitionPosition(job.position, REDEEMED)
        else if permanently failed:
            transitionPosition(job.position, REDEMPTION_EXCEPTION)
            openIncident(job.position, REDEMPTION_FAILURE)
\end{lstlisting}

\subsection{Idempotency and audit rule}

Every command that can move financial state has a deterministic or caller-supplied idempotency key over
\begin{equation}
(\text{account},\text{position},\text{operation},\text{attempt group},\text{external reference}).
\end{equation}
Retries may create new venue order identifiers within one attempt group. Settlement and redemption events use their final venue or chain identifier as unique journal keys. Risk certificates, scenario versions, priority orders, and configuration changes are append-only audit records.

\section{Deterministic Scenario Registry}
\label{app:scenarios}

This appendix fixes every numerical input used by the verifier in \Cref{sec:deterministic}. The registry is author-specified, not calibrated to a venue, and must not be read as observed market data. All financial arithmetic uses decimal rather than binary floating-point values.

\subsection{Robust-sale illustration}

The single-position robust-sale example uses an already effective entry unit cost rather than replaying an entry ask ladder:
\begin{align}
C&=500, & L&=2, & N&=1000,\\
a_{\mathrm{eff}}&=0.42, & q&=N/a_{\mathrm{eff}}=2380.95238095,\\
r&=0.15\text{ per year}, & \Delta t&=7/365,\\
D_u&=500\exp(0.15\cdot7/365)=501.44042702.
\end{align}
The largest registered operating settlement horizon and charge produce
\begin{equation}
\overline H_{u,\Delta}=502.69071323.
\end{equation}
The explicit hard-flat buffer is $m_u=5$, so the lower-envelope target is
\begin{equation}
\overline H_{u,\Delta}+m_u=507.69071323.
\end{equation}

\subsection{Operating settled-proceeds paths}

Each path is a finite gross bid ladder plus a proportional sale fee and settlement charge. The lower envelope is the pointwise minimum of the three settled-proceeds curves.

\begin{table}[H]
\centering
\small
\begin{tabular}{@{}llrrr@{}}
\toprule
Path & Level & Quantity & Gross bid & Sale fee \\
\midrule
Favorable & 1 & 800 & 0.390 & 0.4\% \\
          & 2 & 900 & 0.375 &  \\
          & 3 & 1,200 & 0.355 &  \\
\addlinespace
Base      & 1 & 600 & 0.370 & 1.0\% \\
          & 2 & 700 & 0.355 &  \\
          & 3 & 800 & 0.335 &  \\
          & 4 & 1,000 & 0.300 &  \\
\addlinespace
Adverse operating & 1 & 500 & 0.340 & 1.2\% \\
          & 2 & 700 & 0.320 &  \\
          & 3 & 900 & 0.295 &  \\
          & 4 & 1,000 & 0.260 &  \\
\bottomrule
\end{tabular}
\caption{Registered operating bid ladders. Settlement delays are 5, 30, and 120 seconds; fixed settlement charges are 0.25, 0.60, and 1.25 USDC respectively.}
\end{table}

The resulting planned sale and realized audit quantities are:

\begin{table}[H]
\centering
\small
\begin{tabular}{@{}lrrrrr@{}}
\toprule
Path & Debt & Audit minimum & Planned cap & Residual & Shortfall \\
\midrule
Favorable operating & 501.6904 & 1,311.2140 & 1,606.2949 & 1,069.7384 & 0 \\
Base operating      & 502.0405 & 1,409.2884 & 1,606.2949 & 971.6640 & 0 \\
Adverse operating   & 502.6907 & 1,589.1399 & 1,606.2949 & 791.8125 & 0 \\
\bottomrule
\end{tabular}
\caption{Ex-ante planned cap versus realized post-settlement audit minimum. The planned cap is fixed before the path is known; the audit minimum varies with the realized path.}
\label{tab:deterministic-scenarios}
\end{table}

If the entire planned cap were already matched before cancellation, quantity overshoot relative to the audit minimum would be 295.08087912, 197.00650360, and 17.15501269 tokens respectively. A recourse controller that waits for settlement-confirmed prefixes can reduce overshoot, subject to venue latency and cancellation semantics.

\subsection{Nested uncertainty registry and adversarial transforms}

The deterministic registry labels the three favorable/base/adverse paths as $\Uop$. The registered stress set $\Ustress$ contains all three operating paths plus the transformation and failure fixtures below. The verifier requires $\Uop\subseteq\Ustress$ and treats every $\Ustress\setminus\Uop$ path as non-certified.

For a base bid ladder and reference mid $m_u=0.40$, the transform fixtures are:
\begin{center}
\begin{tabular}{@{}lrrr@{}}
\toprule
Fixture & $k$ & $h$ & $s$ \\
\midrule
Identity & 0 & 0 & 1.00 \\
Top-level removal & 1 & 0 & 1.00 \\
25\% depth haircut & 0 & 0.25 & 1.00 \\
Spread expansion & 0 & 0 & 1.50 \\
Combined adversarial & 1 & 0.35 & 1.50 \\
Complete withdrawal & all & 1.00 & 3.00 \\
\bottomrule
\end{tabular}
\end{center}
For every positive test quantity, the verifier requires the transformed proceeds not to exceed the identity proceeds. Complete withdrawal must produce zero executable quantity and zero proceeds. These are deterministic stress transformations, not empirically calibrated probabilities.

\subsection{Failure-set paths}

The thin-depth failure ladder contains 350 tokens at 0.315 and 350 at 0.275 with a 1.5\% sale fee. Maximum confirmed proceeds are 203.40250000, producing shortfall 299.28821323 against debt 502.69071323. The zero-liquidity path has an empty bid ladder and shortfall 502.69071323. The verifier requires both paths to remain classified as failures; it is an error to label them debt-cleared.

\subsection{Control-capability failure paths}

Four additional negative controls are not price-book scenarios. They represent failures of conditions required by \Cref{thm:robust-clearing}: premature venue close, loss of signing authority, collateral escape from the debt-first control boundary, and persistent settlement failure. For each path, the verifier requires
\begin{equation}
\omega\notin\Uop_u(\Delta)
\qquad\text{and}\qquad
\state{DEBT\_FREE\_SPOT}\text{ is not an admissible ordinary transition}.
\end{equation}
The required outputs are explicit exception states rather than a fabricated successful hard-flat. These controls verify classification logic only; they do not assign probabilities to the failures.

\subsection{Exact leverage-tier illustration}

The leverage-tier test uses a separate finite ask ladder so that entry impact grows with size. Gross ask levels are:

\begin{center}
\begin{tabular}{@{}rrr@{}}
\toprule
Level & Quantity & Gross ask \\
\midrule
1 & 700 & 0.390 \\
2 & 800 & 0.410 \\
3 & 1,000 & 0.445 \\
4 & 1,500 & 0.500 \\
5 & 2,500 & 0.600 \\
\bottomrule
\end{tabular}
\end{center}
The proportional entry fee is 1.0\%. The adverse operating exit curve supplies the robust full-sale proceeds. The buffer is
\begin{equation}
m(L)=2.50+0.005\,LC.
\end{equation}
The finite policy set is
\begin{equation}
\mathcal L_{\mathrm{policy}}
=
\{1,1.25,1.5,1.75,2,2.25,2.5,3\}.
\end{equation}

\begin{table}[H]
\centering
\scriptsize
\begin{tabular}{@{}rrrrrrc@{}}
\toprule
$L$ & Entry quantity & Effective entry cost & Debt upper & Robust proceeds & $\Psi(L)$ & Feasible \\
\midrule
1.00 & 1,241.58415842 & 0.40271132 & 1.25000000 & 401.39211881 & 395.14211881 & yes \\
1.25 & 1,540.02669930 & 0.40583712 & 126.61017831 & 488.37618178 & 356.14100347 & yes \\
1.50 & 1,818.14439871 & 0.41250849 & 251.97035661 & 569.43636645 & 311.21600983 & yes \\
1.75 & 2,096.26209812 & 0.41740964 & 377.33053492 & 650.49655112 & 266.29101620 & yes \\
2.00 & 2,374.37979753 & 0.42116261 & 502.69071323 & 722.06868239 & 211.87796916 & yes \\
2.25 & 2,635.72277228 & 0.42682789 & 628.05089154 & 789.20246574 & 153.02657421 & yes \\
2.50 & 2,883.24752475 & 0.43353891 & 753.41106984 & 852.78662416 & 90.62555432 & yes \\
3.00 & 3,378.29702970 & 0.44401069 & 1,004.13142646 & 857.09000000 & $-157.04142646$ & no \\
\bottomrule
\end{tabular}
\caption{Exact finite-book leverage test. The largest feasible registered tier is $2.50\times$; $3.00\times$ fails despite the scalar benchmark being calculated alongside it.}
\label{tab:leverage-values}
\end{table}

The quantity at $L=2$ differs from the robust-sale illustration because this subsection explicitly walks a finite entry ask ladder, whereas the first subsection fixes an effective entry cost. The two fixtures test different parts of the implementation.

\subsection{Aggregate shared-book registry}

Three positions share one bid ladder:
\begin{center}
\begin{tabular}{@{}rrr@{}}
\toprule
Level & Quantity & Gross bid \\
\midrule
1 & 500 & 0.380 \\
2 & 700 & 0.350 \\
3 & 900 & 0.310 \\
4 & 1,200 & 0.260 \\
\bottomrule
\end{tabular}
\end{center}
The sale fee is 1.0\%. Position inputs are
\begin{center}
\begin{tabular}{@{}lrrrrr@{}}
\toprule
Position & Quantity & Cash & Debt & Buffer & Deadline rank \\
\midrule
P1 & 600 & 0 & 170 & 10 & 3 \\
P2 & 700 & 5 & 190 & 10 & 1 \\
P3 & 850 & 0 & 185 & 10 & 2 \\
\bottomrule
\end{tabular}
\end{center}
The sum of independently re-anchored values is 789.52500000 USDC, while selling the combined quantity through one shared book yields 719.73000000 USDC. The mechanical overstatement is therefore 69.79500000 USDC\@.

The verifier computes $\Gamma_i$, freezes the lowest-coverage-first ordering, and walks the book only once. All three loans must clear in this fixture without reusing depth. A deliberately increased aggregate-debt variant is required to fail. A separate reserve fixture uses shortfalls $(120,80,200)$ and reserve $R=250$, producing pro-rata allocations $(75,50,125)$ and residual shortfalls $(45,30,75)$.

\subsection{Payout, dispute, and redemption grid}

For every successful operating path, the verifier evaluates
\begin{equation}
\pi_h\in\left\{0,\frac12,1\right\},
\qquad
\tau_f-\sigma\in\{0,1,7,30,365\}\text{ days},
\end{equation}
and redemption delays of 0, 1, 7, and 30 days. Loan-channel credit loss remains zero after confirmed debt extinction. User claim value varies linearly with $\pi_h$. Time value and opportunity cost are intentionally excluded from this accounting check.

\subsection{Linear shortfall surface}

The intuition surface in \Cref{fig:shortfall-surface} uses debt-growth factor $\rho_D=1.01$, leverage $L\in[1,4]$, and scalar recovery ratio $g\in[0,1]$:
\begin{equation}
\ell(L,g)=\max\{\rho_D(L-1)-gL,0\}.
\end{equation}
This surface is not used for exact tier admission or reserve sizing.

\subsection{Why no random simulation is included}

A Monte Carlo exercise would require unobserved assumptions about quote cancellation, strategic maker response, early close, settlement failure, signer availability, and correlations. Without venue-specific calibration, such a distribution would create apparent precision rather than evidence. The deterministic registry instead verifies the algebra and makes failure boundaries explicit. A later empirical study can replace the author-specified operating set with pre-registered, data-estimated envelopes.

\section{Reproducibility, Versioning, and Release Package}
\label{app:reproducibility}

\subsection{Package structure}

The complete release is organized as follows:

\begin{lstlisting}[language={}]
axient-debt-free-finality-r0.3.1/
  main.pdf
  main.tex
  main.bbl
  references.bib
  sections/
  appendices/
  figures/
  code/
    generate_deterministic_evaluation.py
  outputs/
    robust_sale_scenarios.csv
    exact_leverage_tiers.csv
    aggregate_liquidity.csv
    stress_transformations.csv
    reserve_waterfall.csv
    control_failure_paths.csv
    verification_report.json
  REVIEW_RESPONSE_r0.3.1.md
  REVISION_NOTES_r0.3.1.md
  README.md
  README_RU.md
  requirements.txt
  ARXIV_SUBMISSION_METADATA.md
  FORESIGHTFLOW_PUBLICATION_PAGE.md
  publication_metadata.json
  abstract.txt
  arxiv_abstract.txt
  CHANGELOG.md
  CITATION.cff
  QA/
  LICENSE
  LICENSES/
  MANIFEST.sha256
\end{lstlisting}

The clean arXiv source archive contains only the manuscript sources, bibliography output, figures required by the manuscript, a short build note, and a cryptographic manifest. The complete package additionally contains verification code, generated outputs, reviewer-response material, metadata, and quality-assurance reports.

\subsection{Reproduction commands}

From the package root:

\begin{lstlisting}[language={}]
python code/generate_deterministic_evaluation.py
pdflatex -interaction=nonstopmode -halt-on-error main.tex
bibtex main
pdflatex -interaction=nonstopmode -halt-on-error main.tex
pdflatex -interaction=nonstopmode -halt-on-error main.tex
\end{lstlisting}

The Python program makes no network request. It regenerates the numerical outputs and all manuscript figures derived from the deterministic fixtures. The LaTeX build consumes only local files.

\subsection{Arithmetic and determinism}

Core financial calculations use Python's \texttt{Decimal} type at precision 50. Binary floating point is used only after values are converted for plotting. Scenario inputs, output tables, and verification results are serialized as decimal strings. The program uses no random seed because the registered scenario set is deterministic.

The current verifier records 249 named checks. The build is successful only if every required check is true. Checks include:
\begin{itemize}
  \item lower-envelope minimality immediately to the left and at the planned solution;
  \item audit feasibility and audit-minimum dominance by the planned cap on every operating path;
  \item quantity conservation;
  \item payout, dispute-duration, and redemption-latency invariance after debt extinction;
  \item failure classification for thin depth and zero liquidity;
  \item exclusion of premature close, signer loss, collateral escape, and persistent settlement failure from the robust operating set;
  \item exact leverage-tier acceptance and rejection;
  \item the linear benchmark identity as a separate algebraic check;
  \item top-of-book double-counting detection;
  \item nested $\Uop\subseteq\Ustress$ classification;
  \item monotonic adversity of registered $T_{k,h,s}$ transformations;
  \item lowest-$\Gamma_i$ aggregate execution without depth reuse;
  \item pro-rata reserve allocation consistency; and
  \item deliberate rejection of an aggregate over-capacity fixture.
\end{itemize}

Negative controls are first-class outputs. A verifier that returns \state{DEBT\_FREE\_SPOT} for the thin-depth or zero-liquidity path, or permits an ordinary debt-free transition after a required control capability has failed, is incorrect even if every operating fixture passes.

\subsection{What the verifier does not establish}

The release does not estimate:
\begin{itemize}
  \item the probability that a real path belongs to $\Uop$;
  \item a venue-specific distribution of settlement latency;
  \item actual-close hazard relative to a published close timestamp;
  \item strategic quote withdrawal caused by the mechanism's own hard-flat flow;
  \item the legal enforceability of a controlled signer or lender lien; or
  \item the economic cost of capital and reserves at production scale.
\end{itemize}
Those quantities require partner data, contractual review, and a separate pre-registered empirical study. They are not inferred from the deterministic fixtures in this paper.

\subsection{Revision discipline}

Version r0.3.1 is a final minor revision following the major r0.3.0 formal response. In addition to retaining all r0.3.0 changes, it:
\begin{enumerate}[label=(\roman*)]
  \item preserves the ex-ante planned-sale and post-settlement audit distinction introduced in r0.3.0;
  \item defines the nested $\Uop\subseteq\Ustress$ taxonomy;
  \item formalizes adversarial $T_{k,h,s}$ book transformations;
  \item fixes the canonical lowest-$\Gamma_i$ priority rule;
  \item adds the pro-rata principal-shortfall reserve allocation;
  \item updates the deterministic verifier for aggregate liquidity and stress transformations; and
  \item states in the abstract that empirical calibration and out-of-sample validation of the uncertainty sets are separate research tasks.
\end{enumerate}
The accompanying response memos distinguish accepted mathematical criticisms from factual API assertions that did not match the cited public documentation and record the final minor-review changes.

Future revisions must classify changes as mathematical correction, policy change, venue-adapter assumption, empirical calibration, implementation/security change, or editorial clarification. A theorem must not be silently modified to fit later data. If empirical evidence rejects the practical operating set, the deployment claim and policy must change while the conditional mathematical statement remains scoped to its assumptions.

\subsection{Suggested empirical protocol}

A venue-specific follow-up should pre-register:
\begin{itemize}
  \item market-selection and sample-adequacy gates;
  \item reconstruction and reconciliation rules for book, match, settlement, and balance data;
  \item the construction and target coverage of $\Uop$, plus the separate $\Ustress$ catalogue;
  \item actual-close and settlement-latency estimators;
  \item fees, gas, retries, and failed-settlement treatment;
  \item position-level and aggregate OI configurations;
  \item deterministic priority or pooling rules;
  \item comparison mechanisms: debt held through finality, full pre-close liquidation, and minimum debt-clearing conversion;
  \item primary outcomes: hard-flat completion, lender shortfall frequency and severity, residual user exposure, reserve draw, and operating-set breach rate; and
  \item falsifiability floors fixed before replay.
\end{itemize}
This follows the pre-registration and negative-result discipline of the earlier ForesightFlow Event-Linked Perpetuals work \citep{nechepurenko2026resolutionaware,nechepurenko2026fillside}.

\subsection{Licensing and archival identity}

The manuscript text is released under CC BY 4.0 and the verification code under the MIT License. The package includes \texttt{CITATION.cff}, arXiv metadata, ForesightFlow publication metadata, and SHA-256 manifests. The project page linked in the manuscript is non-archival; the PDF, source archive, and repository release are the citable research objects.

\section*{Acknowledgments}
The author acknowledges the ForesightFlow research programme for the empirical and theoretical foundation on event-linked perpetuals, informed-flow detection, and prediction-market microstructure on which this mechanism builds. This paper is a standalone mechanism-design study informed by, but not numbered within, the four-paper ForesightFlow Event-Linked Perpetuals programme. The manuscript text is intended for distribution under CC BY 4.0 and the accompanying verification code under the MIT License.

\paragraph{Computational assistance disclosure.}
Automated language and code-assistance tools, including generative language models, were used during manuscript preparation for editorial review, consistency checks, LaTeX typesetting support, and development and testing of deterministic verification scripts. The author specified and reviewed the research question, mechanism, assumptions, formal statements, proofs, numerical configurations, interpretation, and final manuscript, and takes responsibility for the complete work and accompanying materials. These tools were not treated as evidentiary sources and are not credited with authorship.

\bibliographystyle{plainnat}
\bibliography{references}

\end{document}